\documentclass[journal,onecolumn]{IEEEtran}

\usepackage{graphicx}
\usepackage{amsmath}
\usepackage{amssymb}
\usepackage{amsthm}  
\usepackage{mathtools}
\usepackage{booktabs}
\usepackage{array}
\usepackage{multirow}
\usepackage{cite}
\usepackage{makecell}
\usepackage{enumitem}
\usepackage{algorithm}
\usepackage{algpseudocode}
\usepackage{xcolor}
\usepackage{hyperref}
\usepackage{tikz}
\usetikzlibrary{arrows.meta,positioning,shapes.geometric,calc}
\usepackage{pgfplots}
\pgfplotsset{compat=1.17}

\newtheorem{theorem}{Theorem}
\newtheorem{lemma}[theorem]{Lemma}

\newtheorem{proposition}[theorem]{Proposition}
\newtheorem{conjecture}[theorem]{Conjecture}

\theoremstyle{definition}

\newtheorem{construction}[theorem]{Construction}
\newtheorem{example}[theorem]{Example}
\theoremstyle{remark}
\newtheorem{remark}[theorem]{Remark}

\newcommand{\F}{\mathbb{F}}
\newcommand{\Z}{\mathbb{Z}}
\newcommand{\supp}{\mathrm{supp}}
\newcommand{\rank}{\mathrm{rank}}
\newcommand{\wt}{\mathrm{wt}}
\newcommand{\row}{\mathrm{row}}
\newcommand{\invol}[1]{\bar{#1}}
\newcommand{\sinner}[2]{\langle #1,\,#2\rangle_{s}}

\begin{document}
\title{Quantum Bicycle LDPC Codes with High $kd^2/n$ from
Divisor-Driven Search}
\author{Liangdong~Lu,
          Guanmin~Guo, Yang~Liu,
        and~Ruipan~Yang
\thanks{L. Lu, R. Yang, Y. Liu, and G. Guo are with the Department of Basic Science, Air Force Engineering University, Xi'an, Shaanxi,  P. R. China e-mail: (see $kelinglv@163.com,yangruipan@aliyun.com, liu\_yang1@163.com, gmguo\_xjtukgd@yeah.net$).}
}

\maketitle

\begin{abstract}
Bicycle (two-block circulant) quantum low-density parity-check (LDPC)
codes include some of the best known small quantum codes, yet their
design has relied on group-algebra formulations in which the dimension
and distance are accessible only through matrix computation. We show
that in the cyclic case the construction collapses into the polynomial
ring $\F_2[x]/(x^{l}-1)$: self-orthogonality is automatic, the quantum
dimension is read off from a polynomial gcd, and the minimum distance
is certified exactly through the Calderbank correspondence to additive
codes over $\F_4$, turning code search into an algebraically
pre-filtered enumeration that reaches parameter regimes poorly covered
by existing tables. A computer search based on this framework recovers the short codes
$[[42,12,4]]_2$ and $[[62,12,4]]_2$ and produces a family of codes with
competitive figure of merit $kd^2/n$, including $[[66,20,7]]_2$ with
$kd^2/n=14.85$, above the bivariate bicycle code $[[144,12,12]]_2$
($kd^2/n=12$) at less than half the block length, together with
$[[46,2,8]]_2$, $[[66,2,9]]_2$, $[[66,4,8]]_2$, $[[66,6,8]]_2$ and, at
$n=90$, $[[90,16,7]]_2$, $[[90,18,6]]_2$, $[[90,20,6]]_2$. An
exhaustive census at $n=48$ delineates the boundary of this picture:
we exhibit a $[[48,10,6]]_2$ code from a minimal $48$-element group
(the Aydin--Tamo--Barg realization uses $72$ elements), and prove that
distance $5$ forces a stabilizer-rank loss, which excludes
$[[48,10,5]]_2$ from the weight-$8$ symmetric coset family. The
framework thus opens a systematic route to bicycle-type quantum LDPC
codes beyond the reach of group-theoretic searches, and identifies
exactly where genuinely coset-theoretic phenomena begin.
\end{abstract}


\section{Introduction}\label{sec:intro}
\subsection{From surface codes to bicycle codes}
\IEEEPARstart{T}{he} cost of a quantum error-correcting code is ultimately
measured by the overhead it adds to a fault-tolerant computation. Surface
codes~\cite{Kitaev03} remain the reference point: they are geometrically
local and tolerate high error rates, but each patch encodes a single logical
qubit, so the physical-qubit count grows quadratically with the target
distance. Quantum low-density parity-check (LDPC) codes aim to remove this
overhead by keeping every stabilizer measurement at bounded weight while
letting the dimension and distance grow with the block length.

The search for such codes has moved through three methodological phases. The
first was product-based. The hypergraph product codes of Tillich and
Z\'emor~\cite{TZ14} combine two classical codes into a quantum CSS code with
constant rate and distance $\Theta(\sqrt{n})$, and the hyperbicycle codes of
Kovalev and Pryadko~\cite{KP13b} showed that circulant structure is
compatible with finite rate in the same distance regime. The second phase
was asymptotic. Panteleev and Kalachev~\cite{PK22} constructed quantum LDPC
codes with almost linear distance, and the quantum Tanner codes of
Leverrier and Z\'emor~\cite{LZ22}, together with a linear-time
decoder~\cite{LZD22}, established asymptotically good families that decode
efficiently. The third phase, to which the present work belongs, is
group-theoretic. Codes built from two commuting blocks go back to the
sparse-graph constructions of MacKay, Mitchison, and McFadden~\cite{MMM04}
and to the generalized bicycle (GB) codes of Kovalev and
Pryadko~\cite{KP13a}, in which both blocks are circulants. Lin and Pryadko
extended the construction to arbitrary group algebras, giving the two-block
group algebra (2BGA) codes~\cite{2BGA}, and Wang, Lin, and Pryadko analyzed
the underlying commutation mechanism for abelian and non-abelian
groups~\cite{WLP23}. The bivariate bicycle (BB) codes of Bravyi
\emph{et al.}~\cite{BB24} then showed that this line of research has
immediate practical content: the $[[144,12,12]]$ code protects twelve
logical qubits at distance twelve with weight-six checks, reducing the
qubit overhead by roughly an order of magnitude relative to surface codes
of comparable performance, on a layout compatible with superconducting
hardware.

Since the appearance of BB codes, the bicycle family has developed along
four lines. The first is algebraic refinement of the univariate case: small
codes from algebraic extensions of GB codes~\cite{GBExt}, GB codes with
connectivity approaching that of surface codes~\cite{GBLC}, a
classification of the lowest-weight $(2,2)$-GB family~\cite{GB22},
univariate bicycle codes that reduce the two-polynomial search to a single
polynomial~\cite{UB26}, and a cyclic-submodule formulation that exposes the
automorphism group and its fault-tolerant gates~\cite{GBauto}. The second
is multivariate generalization, from multivariate bicycle codes~\cite{MB24}
and independent trivariate bicycle codes~\cite{ITB} to the multivariate
multicycle framework, which unifies these constructions and supports
single-shot decoding~\cite{MM26}; cyclic hypergraph product codes combine
the circulant idea with the older product construction~\cite{CHGP}. The
third concerns logical structure and boundaries: explicit bases of logical
operators and fold-transversal gates~\cite{ES24}, self-dual BB codes with
transversal Clifford gates~\cite{LC25}, and open-boundary planar variants
obtained by pruning~\cite{Prune} or by anyon condensation and lattice
grafting~\cite{Planar}. The fourth is decoding and implementation. BP with
ordered statistics decoding (BP-OSD)~\cite{PK21} remains the standard
benchmark, now joined by almost-linear-time decoders under circuit-level
noise~\cite{BPOTF}, matching decoders that exploit the toric structure of
BB codes~\cite{SWB26, Tan26}, and list decoding of bicycle codes~\cite{ListDec}.
On the hardware side, the locality obstruction bounds of
\cite{Bounds21} have been answered by two-dimensional local
protocols~\cite{Local2D}, morphing circuits that lower the required
connectivity~\cite{Morphing}, multilayer placement and
routing~\cite{HAL}, erasure-biased neutral-atom processors~\cite{Erasure},
and modular architectures~\cite{Modular}.

A second pillar of the present work predates all of the above. The
stabilizer formalism of Calderbank, Rains, Shor, and Sloane
(CRSS)~\cite{CRSS98} provides an algebraic bridge from classical coding
theory to quantum codes: the error group modulo phases is an elementary
abelian $2$-group, operator commutation is captured by a symplectic inner
product, and commuting stabilizer groups correspond to symplectic
self-orthogonal codes, equivalently described as additive self-orthogonal
codes over $\F_4$. This symplectic viewpoint has remained largely disjoint
from the group-algebraic bicycle literature surveyed above; connecting the
two is the starting point of this paper.

\subsection{The problem we address}
The group-theoretic constructions (GB, 2BGA, BB, and the coset-based codes
of Aydin--Tamo--Barg~\cite{ATB}) are powerful, but they are phrased in the
language of group actions, permutation representations, and group algebras.
For code discovery this has two practical consequences. The search space is
indexed by group--subgroup pairs and support sets of group-algebra
elements, which is indirect and hard to prune; and the key quantum
parameters, the dimension $k$ and the distance $d$, cannot be read off
algebraically before large parity-check matrices are built. We therefore
ask: \emph{can the bicycle construction be reformulated entirely in the
symplectic/polynomial domain, so that good quantum LDPC codes are obtained
by analyzing and selecting polynomials, with the dimension and the
self-orthogonality controlled algebraically before any matrix is built?}

\subsection{Contributions}
\begin{enumerate}[leftmargin=2em,label=(\arabic*)]
  \item A polynomial framework with algebraic dimension control. Under the
  cyclic specialization $G=\Z_l$, $H=\{e\}$, circulant blocks act as
  multiplication in $R=\F_2[x]/(x^l-1)$, so the CRSS self-orthogonality
  holds automatically (Theorem~\ref{thm:weld}), and the quantum dimension
  is a polynomial gcd, $k=2\deg\gcd(a,b,x^l-1)$
  (Proposition~\ref{prop:k}), with a two-sided dimension window for the
  lift construction (Theorem~\ref{thm:window}): the dimension is fixed by
  $\deg g$ before any matrix is built.
  \item An exact-distance search that finds high-ratio codes. The quantum
  distance is computed exactly as a two-kernel minimum with the stabilizer
  excluded (Lemma~\ref{lem:exactd}), inside a pipeline that applies three
  polynomial-level filters before building any matrix
  (Construction~\ref{cons:search}, Algorithm~\ref{alg:search}). A Magma
  implementation finds, among others, $[[46,2,8]]_2$, $[[66,2,9]]_2$,
  $[[66,4,8]]_2$, $[[66,6,8]]_2$, $[[66,20,7]]_2$, $[[90,16,7]]_2$,
  $[[90,18,6]]_2$, and $[[90,20,6]]_2$, and recovers the short codes
  $[[42,12,4]]_2$ and $[[62,12,4]]_2$; the $[[66,20,7]]_2$ code attains
  $kd^2/n=14.85$, above the bivariate bicycle code $[[144,12,12]]_2$
  ($kd^2/n=12$) and all trivariate bicycle codes of \cite{ITB}
  (Section~\ref{sec:results}, Tables~\ref{tab:codes}--\ref{tab:explicit}).
  \item A boundary theorem beyond the cyclic case. For the weight-$8$
  symmetric coset-2BGA family at $n=48$, distance $5$ forces a
  stabilizer-rank loss, excluding $[[48,10,5]]_2$ from the family entirely
  (Theorem~\ref{thm:rank-degeneracy}); we also give a $[[48,10,6]]_2$ code
  realized by a minimal $48$-element group (Theorem~\ref{thm:minimal}).
  This identifies precisely which phenomena the polynomial framework
  captures and which are genuinely coset-theoretic
  (Section~\ref{sec:rank48}).
\end{enumerate}

\subsection{Organization}
Section~\ref{sec:prelim} fixes notation and recalls the CRSS and bicycle
constructions. Section~\ref{sec:poly} develops the symplectic polynomial
reformulation. Section~\ref{sec:method} presents the search algorithm.
Section~\ref{sec:results} reports the codes found, and
Section~\ref{sec:decode} benchmarks their BP-OSD decoding performance
against BB reference codes. Section~\ref{sec:rank48} steps outside the
cyclic case and locates the boundary of the polynomial framework through
the $n=48$ coset-2BGA family (minimal group realizations and rank
degeneracy). Section~\ref{sec:value} discusses significance and
limitations, and Section~\ref{sec:concl} concludes.

\section{Preliminaries}\label{sec:prelim}
\subsection{Stabilizer codes and the symplectic inner product}
Let $n$ be the number of qubits. The error group $E_n$ modulo phases is
isomorphic to $\F_2^{2n}$ via $i^{\lambda}X(a)Z(b)\mapsto(a\mid b)$. The
standard symplectic inner product on $\F_2^{2n}$ is
\[
\sinner{(a\mid b)}{(a'\mid b')} = a\cdot b' + a'\cdot b .
\]
Two Pauli operators commute iff their images are symplectically orthogonal.
A subspace $C\le\F_2^{2n}$ is \emph{symplectic self-orthogonal} if
$C\subseteq C_s^{\perp}$, where
$C_s^{\perp}=\{v:\sinner{v}{c}=0,\ \forall c\in C\}$. A stabilizer code with
$C$ self-orthogonal has parameters
\[
[[n,\ k,\ d]]=[[\,n,\ n-\dim C,\ \min\wt_{s}(C_s^{\perp}\setminus C)\,]],
\]
where $\wt_s$ counts coordinates with $(a_i,b_i)\ne(0,0)$.

\subsection{The Calderbank correspondence}
Map $\F_2^{2n}\to\F_4^{n}$ by $\phi(a\mid b)=a+\omega b$, where
$\F_4=\{0,1,\omega,\omega^2\}$, $\omega^2=\omega+1$. Then $C$ is symplectic
self-orthogonal over $\F_2$ iff $\phi(C)$ is additive self-orthogonal over
$\F_4$ under the trace-symplectic form, and the quantum distance is the
minimum weight of $\phi(C_s^{\perp})\setminus\phi(C)$
\cite{CRSS98}. This is how we compute $d$ exactly; excluding the stabilizer
image is essential (Section~\ref{sec:method}).

\subsection{Bicycle (two-block circulant) codes}
A bicycle code is a CSS code with
\[
H_X=[\,A\mid B\,],\qquad H_Z=[\,B^{T}\mid A^{T}\,],
\]
where $A,B$ are $l\times l$ circulant matrices over $\F_2$. This is the
cyclic case of the generalized bicycle codes of \cite{KP13a} and of the
two-block group algebra codes of \cite{2BGA}. Since circulants commute,
$H_XH_Z^{T}=AB+BA=0$ in characteristic $2$, so the CSS orthogonality holds
automatically; the code length is $n=2l$ and the stabilizer row weight is
bounded by $\wt(a)+\wt(b)$, giving the LDPC property when
$\wt(a),\wt(b)=O(1)$.

\section{Symplectic Polynomial Reformulation}\label{sec:poly}
\subsection{From group algebra to quotient ring}
For $G=\Z_l=\langle t\rangle$, $H=\{e\}$, the group algebra is
$\F_2[\Z_l]\cong R:=\F_2[x]/(x^l-1)$. The correspondence is
\begin{center}
\begin{tabular}{@{}ll@{}}
\toprule
Group language & Polynomial language \\
\midrule
group element $t^{i}$ & monomial $x^{i}$ \\
element $a=\sum_i\alpha_i t^{i}$ & $a(x)=\sum_i\alpha_i x^{i}$ \\
left shift $L(t^{i})$ & multiplication by $x^{i}$ \\
permutation matrix $L(a)$ & circulant $\mathrm{circ}(a(x))$ \\
transpose $L(a)^{T}$ & involution $\invol{a}(x)=a(x^{-1})$ \\
weight $|\supp(a)|$ & number of nonzero terms of $a(x)$ \\
\bottomrule
\end{tabular}
\end{center}
Hence $A=\mathrm{circ}(a(x))$, $B=\mathrm{circ}(b(x))$, and
$H_Z=[\,\mathrm{circ}(\invol{b})\mid\mathrm{circ}(\invol{a})\,]$.

The involution interacts with the symplectic form through the following
identity, which will also account for an additional symmetry observed in our
search results (Remark~\ref{rem:cheaporth}).

\begin{lemma}\label{lem:reciprocity}
For every $a(x),b(x)\in R=\F_2[x]/(x^l-1)$,
\[
a(x)\,\invol{b}(x)\;=\;\overline{\,b(x)\,\invol{a}(x)\,},
\]
and consequently $a\invol{b}+b\invol{a}$ is a sum of pairs
$x^{t}+x^{-t}$, i.e.\ it is fixed by the involution term by term.
\end{lemma}
\begin{proof}
The involution $p\mapsto\invol p$ is a ring automorphism of $R$ and an
involution, so
$\overline{b\invol{a}}=\invol{b}\,\overline{\invol{a}}
 =\invol{b}\,a=a\invol{b}$.
\end{proof}

\subsection{Self-orthogonality becomes automatic}
\begin{theorem}\label{thm:weld}
With $C=\row\binom{H_X\ 0}{0\ H_Z}\subset\F_2^{4l}$, the following are
equivalent: (i) $C$ symplectic self-orthogonal; (ii) $H_XH_Z^{T}=0$;
(iii) $AB=BA$; (iv) the polynomial identity
\begin{equation}\label{eq:invol}
a(x)\,b(x)+b(x)\,a(x)\equiv 0 \pmod{x^{l}-1}.
\end{equation}
Since $R$ is commutative, \eqref{eq:invol} holds identically, so the bicycle
structure yields symplectic self-orthogonal codes \emph{for free}.
\end{theorem}
\begin{proof}
(i)$\Leftrightarrow$(ii) from
$\mathcal G\Omega\mathcal G^{T}=\mathrm{diag}(H_XH_Z^{T},H_ZH_X^{T})$.
(ii)$\Leftrightarrow$(iii): $H_XH_Z^{T}=AB+BA$, and char $2$ gives
$H_XH_Z^{T}=0\iff AB=BA$.
(iii)$\Leftrightarrow$(iv): multiplication of circulants corresponds to
polynomial multiplication, so
$AB+BA=\mathrm{circ}(ab+ba)$.
Commutativity of $R$ gives $ab+ba=2ab=0$ in char $2$.
\end{proof}

\begin{remark}\label{rem:cheaporth}
Theorem~\ref{thm:weld} turns commutation checking inside the search loop
from $2l^{2}$ row-pair inner products into \emph{one} polynomial identity
\eqref{eq:invol}, costing one cyclic convolution, $O(l^{2})$ word operations
or $O(l\log l)$ by FFT-style methods, independent of the matrix dimensions;
in the commutative case the identity is void, so self-orthogonality is
certified by the polynomial formalism itself. In our Magma implementation
the orthogonality gate in fact tests the strictly stronger \emph{involution
identity} $a\invol{b}+b\invol{a}\equiv0$, equivalently $AB^{T}+BA^{T}=0$
(i.e.\ $AB^{T}$ symmetric), a condition that random low-weight polynomial
pairs violate with overwhelming probability. Every code in
Tables~\ref{tab:codes}--\ref{tab:newk2} therefore enjoys this additional
symmetry; a structural characterization of the resulting subclass of
bicycle codes is an interesting open problem.
\end{remark}

\subsection{Algebraic dimension formula}
The transpose of a circulant is a circulant, and the map
$a\mapsto\invol{a}$ is an \emph{automorphism} of $R$; hence
$\rank\mathrm{circ}(\invol{a})=\rank\mathrm{circ}(a)$ and
$\rank H_Z=\rank H_X$ identically. The quantum dimension is therefore
governed by a single block rank, which the quotient-ring structure makes
exact:

\begin{proposition}\label{prop:k}
Let $l$ be odd and $a,b\in R=\F_2[x]/(x^l-1)$. Then
\[
\begin{aligned}
k \;&=\; 2l-\rank H_X-\rank H_Z\\
    &=\; 2\deg\gcd\!\big(a(x),\,b(x),\,x^{l}-1\big).
\end{aligned}
\]
\end{proposition}
\begin{proof}
Write $A=\mathrm{circ}(a)$, $B=\mathrm{circ}(b)$. The row space of $A$
inside $R$ (identified with $\F_2^{l}$) is the principal ideal $(a)\subseteq
R$; since $R$ is commutative, the row space of $H_X=[\,A\mid B\,]$ in
$R^2$ is the submodule $\{(sa,sb):s\in R\}$, whose projection to either
coordinate is the ideal
\[
I\;=\;(a,b)\;=\;\bigl(\gcd(a,b,x^{l}-1)\bigr)\;\subseteq\; R .
\]
For $l$ odd, $x^{l}-1$ is squarefree over $\F_2$, so $R$ is a principal
ideal ring and, for $h=\gcd(a,b,x^{l}-1)$,
$\dim_{\F_2}(h)=\deg\!\big((x^{l}-1)/h\big)=l-\deg h$. Projection onto the
first coordinate is injective on the row space: if $sa=0$ in $R$ then
$s\in\mathrm{Ann}_R(a)=((x^{l}-1)/\gcd(a,x^{l}-1))$; since
$h=\gcd(a,b,x^{l}-1)$ divides both $b$ and $\gcd(a,x^{l}-1)$, the annihilator
$((x^{l}-1)/\gcd(a,x^{l}-1))$ is contained in $\mathrm{Ann}_R(b)$, whence
$sb=0$ as well. Therefore $\rank H_X=\dim I=l-\deg h$. Since
$\rank H_Z=\rank H_X$ by the involution symmetry,
\[
k=2l-2\rank H_X=2l-2\bigl(l-\deg h\bigr)=2\deg h. \qedhere
\]
\end{proof}

\begin{remark}\label{rem:prefilter}
Proposition~\ref{prop:k} lets us pre-filter candidates by $k$ using a
polynomial gcd, before forming any matrix. Low-weight random $a,b$ are
generically coprime to $x^l-1$, giving $\rank H_X=l$ and hence $k=0$; to
obtain $k>0$ one must force $a,b$ to share a divisor of $x^l-1$. This
observation drives Construction~\ref{cons:search} below. In the
implementation the gcd filter is the first test applied to every candidate;
only survivors pay for a matrix build, and a rank cross-check
$k\stackrel{?}{=}n-\rank H_X-\rank H_Z$ on those survivors has never
disagreed with the gcd value, as the proposition predicts.
\end{remark}

\section{The Search Method}\label{sec:method}
\subsection{Structured polynomial construction}
\begin{construction}\label{cons:search}
Fix odd $l$ and factor $x^{l}-1=\prod_j f_j(x)$ over $\F_2$. For each divisor
$g=\prod_{j\in S}f_j$ with $\deg g$ in a target range, and for low-weight
$u(x),v(x)\in R$, set
\[
a(x)=g(x)u(x),\qquad b(x)=g(x)v(x)\pmod{x^{l}-1}.
\]
Then $g\mid\gcd(a,b,x^l-1)$, so $k\ge 2\deg g$, with equality when
$\gcd(u,v,x^{l}-1)=1$. The row weight is at most
$\wt(gu)+\wt(gv)$, keeping the code LDPC.
\end{construction}

The construction has a two-sided algebraic handle on $k$:

\begin{theorem}\label{thm:window}
Let $l$ be odd, $g\mid x^{l}-1$, and $u,v\in R$ arbitrary. Put
$a=gu$, $b=gv$ in $R$. Then
\[
2\deg g\;\le\;k\;\le\;2\deg g\;+\;2\deg\gcd\!\bigl(u,\,v,\,\tfrac{x^{l}-1}{g}\bigr).
\]
In particular, if $u$ and $v$ share no cyclotomic factor of $(x^l-1)/g$ then
$k=2\deg g$ exactly, so the divisor degree \emph{dials} the quantum
dimension; and every bicycle code of dimension $k$ admits such a presentation
with $g=\gcd(a,b,x^{l}-1)$, $\deg g=k/2$.
\end{theorem}
\begin{proof}
Since $g\mid x^{l}-1$ and $x^{l}-1$ is squarefree, $\gcd(g,(x^l-1)/g)=1$,
hence
$\gcd(gu,gv,x^{l}-1)=g\cdot\gcd\!\big(u,v,(x^{l}-1)/g\big)$. Apply
Proposition~\ref{prop:k}. For the last claim, given any bicycle code from
$(a,b)$, take $g=\gcd(a,b,x^{l}-1)$, $u=a/g$, $v=b/g$: then
$\gcd(u,v,(x^l-1)/g)=1$ and $k=2\deg g$.
\end{proof}

\begin{lemma}\label{lem:weightcap}
For every $t\in R$, the vector
\[
\lambda(t)\;:=\;\bigl(\,\invol{a}(x)\,t(x)\;\big|\;\invol{b}(x)\,t(x)\,\bigr)
\ \in\ \F_2^{2l}
\]
lies in $\ker H_Z$. Writing
$w_{\min}(a,b)=\min_{t\in R}\bigl(\wt(at)+\wt(bt)\bigr)$ (the minimum is over
the cyclic translates-and-combinations of the pair, and
$\wt(\invol{a}t)=\wt(a\,\invol{t})$), the X-side distance satisfies
\[
\begin{aligned}
d_X\;&\le\;w_{\min}(a,b)\;\le\;\wt(a)+\wt(b)\\
     &\le\;\wt(g)\,\bigl(\wt(u)+\wt(v)\bigr),
\end{aligned}
\]
whenever the module $\{\lambda(t):t\in R\}$ is not entirely contained in
$\row H_X$ (equivalently, whenever the X-stabilizer does not already absorb
the involution-translated pair; this holds generically and is checked
explicitly by the exact distance computation).
\end{lemma}
\begin{proof}
The kernel condition for $(p\mid q)\in\ker H_Z$ is
$B^{T}p+A^{T}q=0$, i.e.\ $\invol{b}\,p+\invol{a}\,q=0$ in $R$. Substituting
$p=\invol{a}t$, $q=\invol{b}t$ gives
$\invol{b}\invol{a}t+\invol{a}\invol{b}t=0$ in characteristic $2$, since $R$
is commutative. The weight chain uses $\wt(\invol{a}t)=\wt(a\invol{t})$
(transpose preserves weight) and the sub-multiplicativity
$\wt(gu)\le\wt(g)\wt(u)$.
\end{proof}

\begin{remark}\label{rem:w2}
Lemma~\ref{lem:weightcap} explains the design constraint
$w_u,w_v\ge3$ used throughout our campaign: for $\wt(u)=\wt(v)=2$ the
translated pair $(\invol{a}t\mid\invol{b}t)$ degenerates, since choosing
$t=\sum_{j=0}^{m-1}x^{ji}$ with $mi\equiv0\pmod l$ annihilates
$u=1+x^{i}$, and the surviving low-weight members of the kernel family
empirically cap $d\le 4$; we never observed a weight-$2$ lift with $d\ge5$.
A sharp closed form for
$w_{\min}(a,b)$ in terms of $(g,u,v)$, in particular whether
$d\le\wt(u)+\wt(v)$ holds when $\gcd(u,v)=1$ (saturated by the
$[[66,20,6]]$ code of Table~\ref{tab:explicit}, where
$\wt(u)+\wt(v)=6=d$), is an interesting open problem; see
Section~\ref{sec:value}.
\end{remark}

\begin{algorithm*}[t]
\caption{Divisor-driven search for bicycle codes}\label{alg:search}
\begin{algorithmic}[1]
\Require $l$-list $\mathcal L$; weight sets $\mathcal W_u,\mathcal W_v$;
range $[k_{\min},k_{\max}]$; $d_{\min}$; enumeration cap $T$
\ForAll{$l\in\mathcal L$}
  \State factor $x^{l}-1=\prod_j f_j$; collect divisors $g=\prod_{j\in S}f_j$
  with $2\deg g\in[k_{\min},k_{\max}]$
  \ForAll{$(w_u,w_v)\in\mathcal W_u\times\mathcal W_v$}
    \For{each candidate pair $(u,v)$ (exhaustive if few, else $T$ samples)}
      \State $a\gets gu$, $b\gets gv$ in $R$;\quad
      $k\gets 2\deg\gcd(a,b,x^l-1)$
      \If{$k\notin[k_{\min},k_{\max}]$} \textbf{continue}\EndIf
      \State dedup on $(a\,|\,b)$
      \If{$a\invol b+b\invol a\not\equiv0$} \textbf{continue}\EndIf
      \State build $H_X=[A|B]$, $H_Z=[B^{T}|A^{T}]$; cross-check $k$
      \If{heuristic distance estimate $<d_{\min}$} \textbf{continue}\EndIf
      \State $d\gets$ exact distance (Lemma~\ref{lem:exactd})
      \If{$d\ge d_{\min}$} store $(g,u,v)$ and parameters\EndIf
    \EndFor
  \EndFor
\EndFor
\end{algorithmic}
\end{algorithm*}

\subsection{Exact distance: excluding the stabilizer}
A subtle point is that the Calderbank image of $C_s^{\perp}$ contains the
stabilizer itself, whose elements can have weight as low as the check
weight; a naive computation of the minimum distance of the additive code
therefore undercounts the quantum distance and rejects genuine
high-distance candidates. The fix is a two-kernel computation:

\begin{lemma}\label{lem:exactd}
For the bicycle code with checks $H_X=[A|B]$, $H_Z=[B^T|A^T]$, define
\[
\begin{aligned}
d_X&=\min\bigl\{\wt(c):c\in\ker H_Z\setminus\row H_X\bigr\},\\
d_Z&=\min\bigl\{\wt(c):c\in\ker H_X\setminus\row H_Z\bigr\}.
\end{aligned}
\]
Then the quantum distance is $d=\min(d_X,d_Z)$, and each side is computed by
walking the weight classes of the kernel code in increasing weight and
testing membership in the stabilizer row space.
\end{lemma}

The weight-enumeration of Lemma~\ref{lem:exactd} is the bottleneck of the
pipeline, so it is guarded by two cheap gates: the gcd pre-filter
(Proposition~\ref{prop:k}) and a Monte-Carlo upper bound on $d$ obtained by
sampling the kernel and reweighting; candidates whose sampled upper bound
already falls below $d_{\min}$ are vetoed before any exact call. For $n>50$,
where Magma's quantum tables are unavailable, the pipeline maintains a
filtering reference $d_{\rm ref}$, the best-known minimum distance of an
$\F_4$ linear code at the near-full dimension $n-k/2$ (BKLC bounds).
Because $n-k/2\approx n$, this reference is deliberately permissive:
$d>d_{\rm ref}$ holds for almost every survivor and carries no information
about how a code compares with the best-known quantum codes at the same
$[[n,k]]_2$. We therefore use $d_{\rm ref}$ only as a search-filter
threshold and never as a performance benchmark.

\section{Results}\label{sec:results}
We report on the sweep
$l\in\{21,23,25,27,31,33,35,39,45\}$,
low-weight lift pairs $(u,v)$ (realized weights $3\le w_u,w_v\le 13$),
$k\in[2,20]$, $d_{\min}=5$, with exhaustive
enumeration of all lift pairs whenever
$\binom{l-1}{w_u-1}\binom{l-1}{w_v-1}\le 2\times10^{5}$ (satisfied at
$l=23$ and, for the lightest weight pairs, at $l\le 35$ up to symmetry)
and $2\times10^{5}$
random draws otherwise. Distances are exact
(Lemma~\ref{lem:exactd}). Table~\ref{tab:codes} collects the best code found
at each parameter set; the quantum-LDPC figure of merit $kd^2/n$ is used
throughout as the quality measure, and for reference the bivariate bicycle
code $[[144,12,12]]_2$ of \cite{BB24} attains $kd^2/n=12$.

\begin{table*}[t]
\centering
\caption{Best bicycle codes found by the polynomial search. Distances
are exact; $w=\wt(a)+\wt(b)$ is the stabilizer row weight. For reference,
the bivariate bicycle code $[[144,12,12]]_2$ of \cite{BB24} attains
$kd^2/n=12$.}
\label{tab:codes}
\renewcommand{\arraystretch}{1.2}
\begin{tabular}{@{}lcccccll@{}}
\toprule
$[[n,k,d]]_2$ & $l$ & $\deg g$ & $(w_u,w_v)$ & $w$ & $kd^2/n$ & note \\
\midrule
$[[42,12,4]]_2$ & 21 & 6  & (2,2)$^{*}$ & --  & 4.57 & weight-2 lift \\
$[[46,2,8]]_2$  & 23 & 1  & (4,4) & 14  & 2.78 & $k=2$ family \\
$[[62,12,4]]_2$ & 31 & 6  & (2,2)$^{*}$ & --  & 3.10 & weight-2 lift \\
$[[66,2,9]]_2$  & 33 & 1  & (4,4) & 16  & 2.45 & $k=2$ family \\
$[[66,4,8]]_2$  & 33 & 2  & (3,4) & 13  & 3.88 & lowest $w$ \\
$[[66,6,8]]_2$  & 33 & 3  & (4,11) & 14  & 5.82 & \\
$[[66,20,7]]_2$ & 33 & 10 & (11,12) & 19  & 14.85 & highest $kd^2/n$ \\
$[[66,20,6]]_2$ & 33 & 10 & (3,3) & 22  & 10.91 & same divisor as $[[66,20,7]]_2$ \\
$[[66,20,5]]_2$ & 33 & 10 & (9,12) & 15  & 7.58 & lowest $w$ at $k=20$ \\
$[[90,16,7]]_2$ & 45 & 8  & (10,3) & 31  & 8.71 & \\
$[[90,18,6]]_2$ & 45 & 9  & (11,13) & 30  & 7.20 & \\
$[[90,20,6]]_2$ & 45 & 8  & (11,3) & 28  & 8.00 & \\
\bottomrule
\end{tabular}

\vspace{2pt}
{\footnotesize $^{*}$Recovered by an earlier weight-2 lift scan; consistent
with Remark~\ref{rem:w2}, weight-2 lifts cap at $d\le4$.}
\end{table*}

\begin{figure*}[t]
\centering
\begin{tikzpicture}
\begin{axis}[width=7.3cm,height=5.1cm,title={$n=66$ ($l=33$)},
  xlabel={dimension $k$}, ylabel={distance $d$},
  xmin=0,xmax=22,ymin=0,ymax=10, xtick={2,4,6,20}, ytick={2,4,6,8},
  grid=both, grid style={gray!25}, legend style={font=\footnotesize},
  legend pos=north west, thick]
  \addplot[only marks,mark=*,mark size=2.6pt,red!75!black,
    point meta=explicit symbolic, nodes near coords,
    every node near coord/.append style={font=\scriptsize, red!75!black}]
    coordinates {(2,9) [2.45] (4,8) [3.88] (6,8) [5.82] (20,7) [14.85]};
  \addlegendentry{this work (exact $d$)}
\end{axis}
\end{tikzpicture}
\hspace{4mm}
\begin{tikzpicture}
\begin{axis}[width=7.3cm,height=5.1cm,title={$n=90$ ($l=45$)},
  xlabel={dimension $k$}, ylabel={distance $d$},
  xmin=14,xmax=22,ymin=0,ymax=10, xtick={16,18,20}, ytick={4,6,8},
  grid=both, grid style={gray!25}, legend style={font=\footnotesize},
  legend pos=north west, thick]
  \addplot[only marks,mark=*,mark size=2.6pt,red!75!black,
    point meta=explicit symbolic, nodes near coords,
    every node near coord/.append style={font=\scriptsize, red!75!black}]
    coordinates {(16,7) [8.71] (18,6) [7.20] (20,6) [8.00]};
  \addlegendentry{this work (exact $d$)}
\end{axis}
\end{tikzpicture}
\caption{Parameter landscape at $n=66$ and $n=90$: exact distance $d$
against dimension $k$ for the best code found at each $k$; labels give
$kd^2/n$ (the bivariate bicycle codes $[[144,12,12]]_2$ and
$[[90,8,10]]_2$ attain $12$ and $8.9$).}
\label{fig:landscape}
\end{figure*}
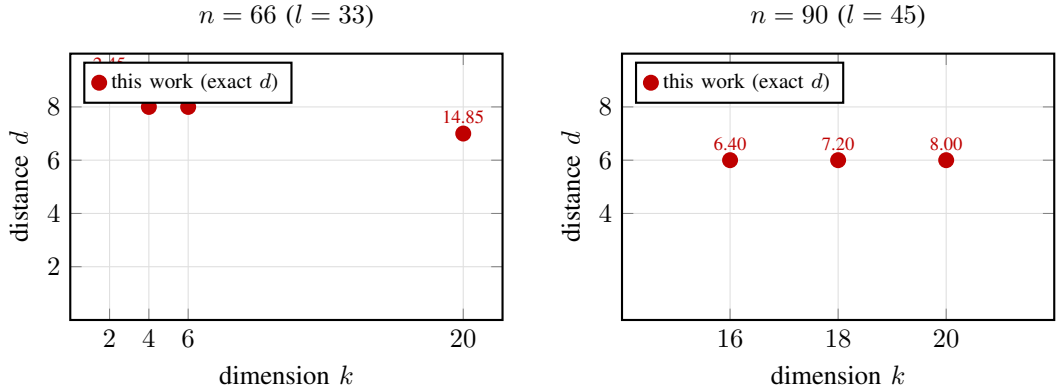

The case $l=33$ is the most productive in the sweep, owing to the
factorization
\[
x^{33}-1=(x+1)(x^{2}+x+1)\,f_{10}\,f'_{10}\,f''_{10}\qquad\text{over }\F_2,
\]
with three distinct irreducible factors of degree $10$
(e.g.\ $f''_{10}=x^{10}+x^{7}+x^{5}+x^{3}+1$). By
Theorem~\ref{thm:window} the dimensions attainable within the searched range
$k\in[2,20]$ are exactly $k\in\{2,4,6,20\}$ (divisor degrees $1,2,3,10$;
degree $11$ would give $k=22$, outside the range), and the search
populates each of them with high-distance codes: the realized distances are
\[
\begin{gathered}
d=5,6,7,8,9\ (k{=}2),\qquad d=5\!-\!8\ (k{=}4),\\
d=5\!-\!8\ (k{=}6),\qquad d=5,6,7\ (k{=}20),
\end{gathered}
\]
with $51$ inequivalent codes in total. The $k=20$ stratum is
notable: a single irreducible divisor
$g=f''_{10}=x^{10}+x^{7}+x^{5}+x^{3}+1$ of degree $10$ supports codes at
all three distances $5,6,7$; the distance-$7$ code attains
$kd^2/n=14.85$, the highest value in the sweep and above the bivariate
bicycle code $[[144,12,12]]_2$ ($kd^2/n=12$).

\begin{example}\label{ex:66}
Take $l=33$, $g=x^{10}+x^{7}+x^{5}+x^{3}+1$ (one of the three degree-$10$
irreducible factors of $x^{33}-1$) and the lifts
\[
\begin{aligned}
u&=\{1,2,3,4,12,13,14,15,16,19,22\},\\
v&=\{0,1,6,8,9,10,11,12,16,19,20,22\}
\end{aligned}
\]
(support notation; $\wt(u)=11$, $\wt(v)=12$). Then
$a=gu$, $b=gv$ have weights $7$ and $12$, so the checks have row weight
$19$; the gcd pre-filter certifies $k=2\deg g=20$ before any matrix is
built, and the exact computation of Lemma~\ref{lem:exactd} gives $d=7$
with $kd^2/n=14.85$, outperforming the bivariate bicycle code
$[[144,12,12]]_2$ ($kd^2/n=12$) at less than half the block length.
\end{example}

Table~\ref{tab:explicit} gives the full construction data of representative
codes, including all codes referenced in this paper's search campaign.

\begin{table*}[t]
\centering
\caption{Explicit polynomial constructions. Supports are
exponents of nonzero terms in $\F_2[x]/(x^{l}-1)$; $a=gu$, $b=gv$;
$w=\wt(a)+\wt(b)$ is the stabilizer row weight; $d_{\rm ref}$ is the
internal filter threshold of Section~\ref{sec:method}, not a performance
benchmark.}
\label{tab:explicit}
\renewcommand{\arraystretch}{1.25}
\footnotesize
\setlength{\tabcolsep}{3.5pt}
\newcommand{\suppS}[1]{$\{\suppSlist#1,\suppSend\}$}
\def\suppSlist#1,#2\suppSend{#1\ifx\relax#2\relax\else,\allowbreak\ \suppSlist#2\suppSend\fi}
\begin{tabular}{@{}p{1.75cm}p{0.55cm}p{2.3cm}p{2.6cm}p{3.3cm}ccp{0.8cm}@{}}
\toprule
$[[n,k,d]]_2$ & $l$ & $\supp g$ & $\supp u$ & $\supp v$ & $(w_u,w_v)$ & $w$ & $d_{\rm ref}$ \\
\midrule
$[[46,2,8]]_2$ & 23
  & \suppS{0,1}
  & \suppS{0,12,17,18}
  & \suppS{0,4,8,12}
  & (4,4) & 14 & 2 \\
\addlinespace[2pt]
$[[66,2,9]]_2$ & 33
  & \suppS{0,1}
  & \suppS{0,4,21,25}
  & \suppS{0,5,20,25}
  & (4,4) & 16 & 2 \\
\addlinespace[2pt]
$[[66,4,8]]_2$ & 33
  & \suppS{0,1,2}
  & \suppS{0,1,17}
  & \suppS{0,1,4,30}
  & (3,4) & 13 & 2 \\
\addlinespace[2pt]
$[[66,6,8]]_2$ & 33
  & \suppS{0,3}
  & \suppS{0,5,14,23}
  & \suppS{1,2,3,6,9,12,15,18,21,24,27}
  & (4,11) & 14 & 2 \\
\addlinespace[2pt]
$[[66,20,5]]_2$ & 33
  & \suppS{0,3,5,7,10}
  & \suppS{1,3,4,8,9,11,12,14,19}
  & \suppS{0,1,3,4,7,9,10,14,15,16,18,21}
  & (9,12) & 15 & 5 \\
\addlinespace[2pt]
$[[66,20,6]]_2$ & 33
  & \suppS{0,3,5,7,10}
  & \suppS{0,1,7}
  & \suppS{0,2,16}
  & (3,3) & 22 & 5 \\
\addlinespace[2pt]
$[[66,20,7]]_2$ & 33
  & \suppS{0,3,5,7,10}
  & \suppS{1,2,3,4,12,13,14,15,16,19,22}
  & \suppS{0,1,6,8,9,10,11,12,16,19,20,22}
  & (11,12) & 19 & 5 \\
\addlinespace[2pt]
$[[90,16,7]]_2$ & 45
  & \suppS{0,1,3,4,5,7,8}
  & \suppS{5,6,10,15,16,21,25,30,31,36}
  & \suppS{0,18,27}
  & (10,3) & 31 & 4 \\
\addlinespace[2pt]
$[[90,18,6]]_2$ & 45
  & \suppS{0,3,4,5,8,9}
  & \suppS{0,2,3,5,9,17,18,20,32,33,35}
  & \suppS{0,2,3,5,9,17,18,20,24,30,32,33,35}
  & (11,13) & 30 & 5 \\
\addlinespace[2pt]
$[[90,20,6]]_2$ & 45
  & \suppS{0,1,3,4,5,7,8}
  & \suppS{4,5,6,11,15,20,21,26,30,35,36}
  & \suppS{0,14,31}
  & (11,3) & 28 & 5 \\
\bottomrule
\end{tabular}
\end{table*}

Extending the search to lengths $n\le 200$ (i.e.\ $l\le 100$, including
$l=85$, $n=170$) and to a wider range of asymmetric weight pairs $(w_u,w_v)$
yields a substantially larger catalogue of codes with
competitive parameters. Table~\ref{tab:newparam} collects the codes with
the highest ratio $kd^2/n$; Table~\ref{tab:newdecode} lists the low-weight
codes selected for the decoding study of Section~\ref{sec:decode}; and
Table~\ref{tab:newk2} exhibits the $k=2$ family (degree-one divisors
$g=x+1$) whose distance grows roughly with the block length while the stabilizer
weight stays at most $16$.

\begin{table*}[t]
\centering
\caption{High-$kd^2/n$ codes from the polynomial search. Supports are exponents of nonzero terms in $\F_2[x]/(x^l-1)$; $a=gu$, $b=gv$; $w=\wt(a)+\wt(b)$.}
\label{tab:newparam}
\renewcommand{\arraystretch}{1.25}
\footnotesize
\setlength{\tabcolsep}{3.5pt}
\newcommand{\suppS}[1]{$\{\suppSlist#1,\suppSend\}$}
\def\suppSlist#1,#2\suppSend{#1\ifx\relax#2\relax\else,\allowbreak\ \suppSlist#2\suppSend\fi}
\begin{tabular}{@{}p{1.75cm}p{0.55cm}p{2.4cm}p{2.9cm}p{2.9cm}cc@{}}
\toprule
$[[n,k,d]]_2$ & $l$ & $\supp g$ & $\supp a$ & $\supp b$ & $w$ & $kd^2/n$ \\
\midrule
$[[66,20,7]]_2$ & 33
  & \suppS{0,3,5,7,10}
  & \suppS{1,2,3,5,10,27,32}
  & \suppS{0,1,3,4,5,7,13,15,22,24,30,32}
  & 19 & 14.85 \\
\addlinespace[2pt]
$[[42,16,6]]_2$ & 21
  & \suppS{0,1,3,6,8}
  & \suppS{0,1,2,8,9,11,17,18,20}
  & \suppS{0,1,2,3,4,6,8,11,12,14,18,19,20}
  & 22 & 13.71 \\
\addlinespace[2pt]
$[[90,18,8]]_2$ & 45
  & \suppS{0,3,4,5,8,9}
  & \suppS{0,1,2,3,4,6,8,9,12,15,16,17,20,21,36,39,40,41,42,44}
  & \suppS{0,1,2,3,4,5,8,9,27,30,31,32,34,35,36,37,39,41}
  & 38 & 12.80 \\
\addlinespace[2pt]
$[[42,14,6]]_2$ & 21
  & \suppS{0,1,2,4,6}
  & \suppS{1,3,5,6,7,9,12,13,14,16,17,19}
  & \suppS{1,2,4,6,11,12,13,14,15,18,19,20}
  & 24 & 12.00 \\
\addlinespace[2pt]
$[[70,16,7]]_2$ & 35
  & \suppS{0,1,3,5,6,8}
  & \suppS{0,1,3,5,6,8,10,11,13,15,16,18,21,23,24,25,26,28,29,30}
  & \suppS{1,2,3,5,6,8,11,13,14,15,16,18,19,20,29,30,32,34}
  & 38 & 11.20 \\
\addlinespace[2pt]
$[[90,20,7]]_2$ & 45
  & \suppS{0,1,2,4,5,8,10}
  & \suppS{0,1,3,5,7,8,9,10,36,39,40,42}
  & \suppS{0,1,2,4,5,8,10,21,22,23,25,26,29,30,32,34,35,38,40}
  & 31 & 10.89 \\
\addlinespace[2pt]
$[[42,18,5]]_2$ & 21
  & \suppS{0,2,3,4,6,7,8,9}
  & \suppS{0,1,2,5,7,10,12,13,15,16,18,19}
  & \suppS{1,3,4,6,8,9,10,12,16,20}
  & 22 & 10.71 \\
\addlinespace[2pt]
$[[54,16,6]]_2$ & 27
  & \suppS{0,3,6}
  & \suppS{0,1,2,3,4,5,6,8,25}
  & \suppS{0,1,3,5,6,8,11,22,25}
  & 18 & 10.67 \\
\addlinespace[2pt]
$[[42,12,6]]_2$ & 21
  & \suppS{0,1,2}
  & \suppS{1,2,13,14,15,17,18,20}
  & \suppS{0,1,2,5,6,7,8,9,10,16,17,18}
  & 20 & 10.29 \\
\addlinespace[2pt]
$[[70,20,6]]_2$ & 35
  & \suppS{0,2,4,5,6,8,10}
  & \suppS{1,2,3,4,5,6,7,8,9,11,13,15,30,32,34}
  & \suppS{0,2,5,8,10,11,34}
  & 22 & 10.29 \\
\addlinespace[2pt]
$[[170,18,8]]_2$ & 85
  & \suppS{0,3,4,5,6,7,8}
  & \suppS{0,3,4,5,6,7,8,46,49,50,51,52,53,54,59,62,63,64,65,66,67,72,75,76,77,78,79,80}
  & \suppS{0,3,4,5,6,7,8,46,49,50,51,52,53,54,61,64,65,66,67,68,69,70,73,74,75,76,77,78}
  & 56 & 6.78 \\
\addlinespace[2pt]
$[[170,18,7]]_2$ & 85
  & \suppS{0,1,3,9}
  & \suppS{0,1,3,9,22,23,25,31,38,39,41,47,69,70,72,78}
  & \suppS{0,1,3,9,11,12,14,20,22,23,25,31}
  & 28 & 5.19 \\
\addlinespace[2pt]
$[[170,16,7]]_2$ & 85
  & \suppS{0,1,5,7,8}
  & \suppS{0,1,6,7,8,10,12,13,45,46,50,52,53}
  & \suppS{0,1,6,7,8,10,12,13,22,23,27,29,30,68,69,73,75,76}
  & 31 & 4.61 \\
\addlinespace[2pt]
$[[170,18,6]]_2$ & 85
  & \suppS{0,1,3,9}
  & \suppS{0,1,3,9,39,40,42,46,47,48,49,55}
  & \suppS{0,1,3,9,31,32,34,40,54,55,57,63}
  & 24 & 3.81 \\
\addlinespace[2pt]
$[[170,10,8]]_2$ & 85
  & \suppS{0,5}
  & \suppS{0,1,3,5,8,78,81,83}
  & \suppS{0,1,10,76}
  & 12 & 3.76 \\
\bottomrule
\end{tabular}

\vspace{2pt}
{\footnotesize The code $[[66,20,7]]_2$ reaches $kd^2/n=14.85$, outperforming the bivariate bicycle code $[[144,12,12]]_2$ ($kd^2/n=12$) and all trivariate bicycle codes of \cite{ITB} ($kd^2/n\le 10$) at less than half the block length.}
\end{table*}

\begin{table*}[t]
\centering
\caption{Codes selected for decoding simulations: stabilizer row weight $w\le 20$ and distance $d\ge 6$.}
\label{tab:newdecode}
\renewcommand{\arraystretch}{1.25}
\footnotesize
\setlength{\tabcolsep}{3.5pt}
\newcommand{\suppS}[1]{$\{\suppSlist#1,\suppSend\}$}
\def\suppSlist#1,#2\suppSend{#1\ifx\relax#2\relax\else,\allowbreak\ \suppSlist#2\suppSend\fi}
\begin{tabular}{@{}p{1.75cm}p{0.55cm}p{2.4cm}p{2.9cm}p{2.9cm}cc@{}}
\toprule
$[[n,k,d]]_2$ & $l$ & $\supp g$ & $\supp a$ & $\supp b$ & $w$ & $kd^2/n$ \\
\midrule
$[[42,6,7]]_2$ & 21
  & \suppS{0,1,2}
  & \suppS{0,1,2,3,4,5,14,17}
  & \suppS{0,2,11,12,13,19}
  & 14 & 7.00 \\
\addlinespace[2pt]
$[[54,6,8]]_2$ & 27
  & \suppS{0,3}
  & \suppS{0,1,3,4,10,13,20,23}
  & \suppS{0,2,3,6,9,19,22,26}
  & 16 & 7.11 \\
\addlinespace[2pt]
$[[66,6,8]]_2$ & 33
  & \suppS{0,3}
  & \suppS{0,3,5,8,14,17,23,26}
  & \suppS{1,2,3,4,5,30}
  & 14 & 5.82 \\
\addlinespace[2pt]
$[[78,6,8]]_2$ & 39
  & \suppS{0,3}
  & \suppS{1,2,3,36,37,38}
  & \suppS{3,15,18,21,24,36}
  & 12 & 4.92 \\
\addlinespace[2pt]
$[[90,4,10]]_2$ & 45
  & \suppS{0,1,2}
  & \suppS{1,3,4,43,44}
  & \suppS{0,1,2,7,8,9,38,39,40}
  & 14 & 4.44 \\
\addlinespace[2pt]
$[[90,2,10]]_2$ & 45
  & \suppS{0,1}
  & \suppS{1,27,29,44}
  & \suppS{0,1,4,5,23,24,27,28}
  & 12 & 2.22 \\
\addlinespace[2pt]
$[[90,10,8]]_2$ & 45
  & \suppS{0,3,4}
  & \suppS{0,4,6,7,20,23,24,28,31,32}
  & \suppS{0,4,6,7,19,22,23,29,32,33}
  & 20 & 7.11 \\
\addlinespace[2pt]
$[[54,16,6]]_2$ & 27
  & \suppS{0,3,6}
  & \suppS{0,1,2,3,4,5,6,8,25}
  & \suppS{0,1,3,5,6,8,11,22,25}
  & 18 & 10.67 \\
\addlinespace[2pt]
$[[170,10,8]]_2$ & 85
  & \suppS{0,5}
  & \suppS{0,1,3,5,8,78,81,83}
  & \suppS{0,1,10,76}
  & 12 & 3.76 \\
\addlinespace[2pt]
$[[170,10,6]]_2$ & 85
  & \suppS{0,5}
  & \suppS{0,5,11,16,74,79}
  & \suppS{0,5,27,32,58,63}
  & 12 & 2.12 \\
\addlinespace[2pt]
$[[170,2,6]]_2$ & 85
  & \suppS{0,1}
  & \suppS{0,1,41,42,44,45}
  & \suppS{2,84}
  & 8 & 0.42 \\
\addlinespace[2pt]
$[[170,10,7]]_2$ & 85
  & \suppS{0,5}
  & \suppS{0,1,5,72,76,77}
  & \suppS{0,5,36,41,72,77}
  & 12 & 2.88 \\
\bottomrule
\end{tabular}

\vspace{2pt}
{\footnotesize Among these, $[[90,2,10]]_2$ ($w=12$) and $[[78,6,8]]_2$ ($w=12$) have the smallest stabilizer weights; $[[170,2,6]]_2$ has $w=8$, the smallest weight found at length $n=170$.}
\end{table*}

\begin{table*}[t]
\centering
\caption{The $k=2$ family: distance $d\ge 8$ and stabilizer row weight $w\le 16$ at every length $n=2l$, $l=23,25,27,33,35,39,45$.}
\label{tab:newk2}
\renewcommand{\arraystretch}{1.25}
\footnotesize
\setlength{\tabcolsep}{3.5pt}
\newcommand{\suppS}[1]{$\{\suppSlist#1,\suppSend\}$}
\def\suppSlist#1,#2\suppSend{#1\ifx\relax#2\relax\else,\allowbreak\ \suppSlist#2\suppSend\fi}
\begin{tabular}{@{}p{1.75cm}p{0.55cm}p{2.4cm}p{2.9cm}p{2.9cm}cc@{}}
\toprule
$[[n,k,d]]_2$ & $l$ & $\supp g$ & $\supp a$ & $\supp b$ & $w$ & $kd^2/n$ \\
\midrule
$[[46,2,8]]_2$ & 23
  & \suppS{0,1}
  & \suppS{0,1,12,13,17,19}
  & \suppS{0,1,4,5,8,9,12,13}
  & 14 & 2.78 \\
\addlinespace[2pt]
$[[50,2,9]]_2$ & 25
  & \suppS{0,1}
  & \suppS{0,1,10,11,14,15,21,22}
  & \suppS{0,1,10,11,16,17,19,20}
  & 16 & 3.24 \\
\addlinespace[2pt]
$[[54,2,9]]_2$ & 27
  & \suppS{0,1}
  & \suppS{0,1,2,3,13,14,16,17}
  & \suppS{0,1,2,3,10,11,19,20}
  & 16 & 3.00 \\
\addlinespace[2pt]
$[[66,2,9]]_2$ & 33
  & \suppS{0,1}
  & \suppS{0,1,4,5,21,22,25,26}
  & \suppS{0,1,5,6,20,21,25,26}
  & 16 & 2.45 \\
\addlinespace[2pt]
$[[70,2,10]]_2$ & 35
  & \suppS{0,1}
  & \suppS{0,1,3,4,11,12,14,15}
  & \suppS{0,1,14,15,18,19,31,32}
  & 16 & 2.86 \\
\addlinespace[2pt]
$[[78,2,9]]_2$ & 39
  & \suppS{0,1}
  & \suppS{0,1,16,18,33,34}
  & \suppS{0,1,14,15,19,20,33,34}
  & 14 & 2.08 \\
\addlinespace[2pt]
$[[90,2,10]]_2$ & 45
  & \suppS{0,1}
  & \suppS{1,27,29,44}
  & \suppS{0,1,4,5,23,24,27,28}
  & 12 & 2.22 \\
\bottomrule
\end{tabular}

\vspace{2pt}
{\footnotesize All codes share $\deg g=1$ ($k=2$); the distance grows roughly with the length $n$ while the stabilizer weight stays at most $16$.}
\end{table*}

\begin{remark}\label{rem:largek}
The $[[90,18,6]]_2$ and $[[90,20,6]]_2$ codes have dimension well beyond
the $k=8$ of the BB code $[[90,8,10]]_2$ at the same length, and the
$[[66,20,7]]_2$ code shows the same behavior at $n=66$: the large-$k$,
moderate-$d$ region is under-sampled by group-theoretic searches that
restrict to symmetric or regular weights. The polynomial search reaches it
naturally because $k$ is set directly by $\deg g$ (Theorem~\ref{thm:window})
while the distance is protected by the lift weights
(Lemma~\ref{lem:weightcap}).
\end{remark}

\begin{remark}
Within the $l=33$, $k=20$ family, all found codes share the same divisor
$g$; the distance is therefore entirely a function of the lifts
$(u,v)$. We observe that the maximal distance $d=7$ is attained by lifts of
weight $(11,12)$ whose products with $g$ cancel extensively, down to
$(\wt(a),\wt(b))=(7,12)$, while the sparsest lifts cap at $d=6$
(Table~\ref{tab:codes}). This suggests that the relevant
quantity is the reduced weight profile of $(gu,gv)$, not the nominal
lift weights; cf.\ Lemma~\ref{lem:weightcap}.
\end{remark}

The BP-OSD decoding performance of the selected low-weight codes is
reported in Section~\ref{sec:decode}.

\section{Decoding Performance}\label{sec:decode}
We benchmark the finite-length error-correction performance of the
low-weight codes of Table~\ref{tab:newdecode} under the standard
code-capacity noise model: each physical qubit independently
suffers a depolarizing error with probability $p$, and syndrome
measurements are assumed perfect. This is the benchmark setting of
\cite{BB24,ITB} for comparing code families without circuit-level
implementation details. The decoder is belief propagation with ordered
statistics decoding (BP-OSD) \cite{PK21}, the decoder class used for the
trivariate bicycle codes of \cite{ITB}; we use the sum-product update
with up to $50$ iterations followed by an OSD combination-sweep of order
$1$. Because the codes are CSS, $X$-type and $Z$-type errors are decoded
independently: the $X$-syndrome $s_X=H_Z e_X$ is decoded against the
check matrix $H_Z$ and likewise for $Z$, with single-sided depolarizing
rate $q=2p/3$. A trial is declared a failure if the residual error of
either sector lies outside the corresponding stabilizer row space, i.e.\
if the decoder output differs from the actual error by a nontrivial
logical operator. For reference we simulate, with the same decoder
configuration, the two weight-$6$ bivariate bicycle (BB) codes of
\cite{BB24}: $[[72,12,6]]$ and $[[144,12,12]]$, on $6\times6$ and
$12\times6$ tori with $A=x^3+y+y^2$ and $B=y^3+x+x^2$, respectively.

\subsection{Logical error rates}
Figures~\ref{fig:decode_long} and~\ref{fig:decode_short} show the
per-shot logical error rate $p_L$ as a function of the physical error
rate $p$, obtained by Monte Carlo simulation ($300$--$3000$ trials per
point, with up to $8000$ trials at the lowest error rates). The codes
are grouped by block length: Figure~\ref{fig:decode_long} compares the
seven codes with $n\ge 90$ against the distance-$12$ reference
$[[144,12,12]]$; Figure~\ref{fig:decode_short} compares the five codes
with $n\le 78$ against $[[72,12,6]]$. For each code of distance $d$, the
sub-threshold data is fit to the heuristic curve
\begin{equation}\label{eq:pLfit}
\tilde p_L(p)=p^{d/2}\, e^{\alpha+\beta p+\gamma p^2},
\end{equation}
following \cite{BB24,ITB}; the fit parameters are listed in
Table~\ref{tab:fit}. All codes exhibit a clear waterfall region. Among
the $n\ge 90$ codes, the steepest suppression is attained by the
distance-$10$ codes $[[90,2,10]]_2$ and $[[90,4,10]]_2$: at $p=10^{-3}$
their extrapolated logical error rates drop to $p_L\approx
2.2\times10^{-7}$, roughly four orders of magnitude below the physical
error rate, consistent with the $p^{d/2}$ scaling of \eqref{eq:pLfit}.
Even the shortest code, $[[42,6,7]]_2$, reaches $p_L\approx
7.8\times10^{-6}$ at $p=10^{-3}$, already more than two orders of
magnitude of suppression with only $42$ physical qubits. The long code
$[[170,10,8]]_2$ ($w=12$) shows no logical failure at $p\le 2\%$ within
the sampled statistics.

As expected from their low stabilizer weight $w=6$, the BB reference
codes exhibit the steepest sub-threshold slopes: at $p=10^{-3}$,
$[[144,12,12]]$ reaches $p_L\approx 4\times10^{-10}$, about three orders
of magnitude below the best $n\le 90$ code of
Table~\ref{tab:newdecode}, and $[[72,12,6]]$ reaches $p_L\approx
6.7\times10^{-7}$, on par with the shortest codes of
Figure~\ref{fig:decode_short}. All curves in
Figures~\ref{fig:decode_long} and~\ref{fig:decode_short} were obtained
with the same BP-OSD configuration, so the comparison is decoder-fair.

\subsection{Pseudo-thresholds}
Following \cite{BB24}, the pseudo-threshold $p_0$ of each code is defined
as the solution of the break-even equation
\begin{equation}\label{eq:break}
p_L(p_0)=1-(1-p_0)^{k},
\end{equation}
the probability that at least one of the $k$ unencoded qubits suffers an
error. Table~\ref{tab:pseudo} reports $p_0$ together with the
extrapolated logical error rates at $p=10^{-3}$ and $p=10^{-2}$ for the
codes of Table~\ref{tab:newdecode} and for the two BB reference codes.
Among the new codes, $p_0$ ranges from $3.1\%$ ($[[54,6,8]]_2$) to
$11.9\%$ ($[[54,16,6]]_2$); the BB reference $[[144,12,12]]$ attains
$p_0\approx 12.7\%$, while the measured $p_L$ of $[[72,12,6]]$ remains
below the break-even value at every sampled error rate ($p\le 0.13$); the
$[[170,10,8]]_2$ code likewise stays below break-even throughout the sampled
range (no logical failure was observed at $p\le2\%$), so no $p_0$ is
reported for it either.
The higher-$k$ codes trade pseudo-threshold for encoding rate, as
expected from \eqref{eq:break}.

These results confirm that the stabilizer row weight $w$ is the primary
driver of decoding performance, in agreement with the observation of
\cite{ATB} that doubling the check weight from $6$ to $8$ roughly halves
the circuit-level threshold. The codes of Table~\ref{tab:newdecode},
with $8\le w\le 20$, are therefore the natural candidates for
circuit-level fault-tolerance studies; in particular, the smallest
weights $w=8$ ($[[170,2,6]]_2$) and $w=12$
($[[90,2,10]]_2$, $[[78,6,8]]_2$, $[[170,10,8]]_2$) are expected to
admit the deepest sub-threshold suppression per round.

\begin{table}[t]
\centering
\caption{Pseudo-thresholds and extrapolated logical error rates under
code-capacity depolarizing noise (BP-OSD, $50$ iterations, OSD-CS order
$1$). $p_0$ solves $p_L(p_0)=1-(1-p_0)^k$; $p_L$ is extrapolated from
Table~\ref{tab:fit}. The last two rows are the bivariate bicycle
reference codes of \cite{BB24}, simulated with the same decoder.}
\label{tab:pseudo}
\renewcommand{\arraystretch}{1.2}
\begin{tabular}{@{}lccccc@{}}
\toprule
$[[n,k,d]]_2$ & $w$ & $p_0$ & $p_L(10^{-3})$ & $p_L(10^{-2})$ \\
\midrule
$[[42,6,7]]_2$ & 14 & 0.1041 & $7.8\times10^{-6}$ & $1.1\times10^{-2}$ \\
$[[54,6,8]]_2$ & 16 & 0.0306 & $3.2\times10^{-6}$ & $1.3\times10^{-2}$ \\
$[[66,6,8]]_2$ & 14 & 0.1098 & $2.8\times10^{-6}$ & $9.8\times10^{-3}$ \\
$[[78,6,8]]_2$ & 12 & 0.1124 & $8.9\times10^{-7}$ & $3.6\times10^{-3}$ \\
$[[90,4,10]]_2$ & 14 & 0.1041 & $2.2\times10^{-7}$ & $5.5\times10^{-3}$ \\
$[[90,2,10]]_2$ & 12 & 0.0939 & $2.2\times10^{-7}$ & $5.1\times10^{-3}$ \\
$[[90,10,8]]_2$ & 20 & 0.0988 & $2.9\times10^{-6}$ & $1.2\times10^{-2}$ \\
$[[54,16,6]]_2$ & 18 & 0.1187 & $6.9\times10^{-5}$ & $3.5\times10^{-2}$ \\
$[[170,10,8]]_2$ & 12 & -- & $6.8\times10^{-9}$ & $8.1\times10^{-5}$ \\
$[[170,10,6]]_2$ & 12 & 0.1081 & $9.3\times10^{-6}$ & $6.0\times10^{-3}$ \\
$[[170,2,6]]_2$ & 8 & 0.0346 & $5.7\times10^{-6}$ & $4.0\times10^{-3}$ \\
$[[170,10,7]]_2$ & 12 & 0.1103 & $7.8\times10^{-6}$ & $1.2\times10^{-2}$ \\
\midrule
\multicolumn{5}{@{}l}{\footnotesize BB reference codes of \cite{BB24}:}\\
$[[72,12,6]]_2$ & 6 & -- & $6.7\times10^{-7}$ & $7.2\times10^{-4}$ \\
$[[144,12,12]]_2$ & 6 & 0.1273 & $4.0\times10^{-10}$ & $1.0\times10^{-4}$ \\
\bottomrule
\end{tabular}
\end{table}
\begin{table}[t]
\centering
\caption{Fit parameters of $\tilde p_L(p)=p^{d/2}\,e^{\alpha+\beta p+\gamma p^2}$
for the curves of Figures~\ref{fig:decode_long} and~\ref{fig:decode_short}
(sub-threshold data $p\ge 0.005$).}
\label{tab:fit}
\renewcommand{\arraystretch}{1.2}
\begin{tabular}{@{}lccccc@{}}
\toprule
$[[n,k,d]]_2$ & $d$ & $\alpha$ & $\beta$ & $\gamma$ \\
\midrule
$[[42,6,7]]_2$ & 7 & 12.50 & -88 & 359 \\
$[[54,6,8]]_2$ & 8 & 15.10 & -109 & 427 \\
$[[66,6,8]]_2$ & 8 & 14.96 & -122 & 548 \\
$[[78,6,8]]_2$ & 8 & 13.80 & -105 & 481 \\
$[[90,4,10]]_2$ & 10 & 19.38 & -162 & 713 \\
$[[90,2,10]]_2$ & 10 & 19.37 & -170 & 762 \\
$[[90,10,8]]_2$ & 8 & 14.96 & -97 & 355 \\
$[[54,16,6]]_2$ & 6 & 11.21 & -79 & 312 \\
$[[170,10,8]]_2$ & 8 & 8.81 & 21 & -205 \\
$[[170,10,6]]_2$ & 6 & 9.19 & -51 & 227 \\
$[[170,2,6]]_2$ & 6 & 8.69 & -43 & 153 \\
$[[170,10,7]]_2$ & 7 & 12.50 & -87 & 367 \\
\midrule
\multicolumn{5}{@{}l}{\footnotesize BB reference codes of \cite{BB24}:}\\
$[[72,12,6]]_2$ & 6 & 6.50 & 8 & -104 \\
$[[144,12,12]]_2$ & 12 & 19.97 & -162 & 790 \\
\bottomrule
\end{tabular}
\end{table}
\begin{figure*}[t]
\centering
\begin{tikzpicture}
\begin{loglogaxis}[width=12.5cm,height=8.2cm,
  xlabel={physical error rate $p$}, ylabel={logical error rate $p_L$},
  xmin=0.001,xmax=0.2,ymin=1e-8,ymax=1.4,
  grid=both, grid style={gray!18},
  tick label style={font=\small}, label style={font=\small},
  legend style={font=\footnotesize, at={(1.03,0.5)}, anchor=west, draw=gray!40},
  legend cell align=left, thick, mark size=2.4pt]
  \addplot[red!75!black,mark=*,mark options={fill=white}] coordinates {(0.005,0.001) (0.01,0.006666666666666667) (0.02,0.01125) (0.03,0.032) (0.05,0.1175) (0.07,0.2075) (0.1,0.3775) (0.13,0.6)};
  \addlegendentry{[[90,2,10]]}
  \addplot[red!75!black,dashed,mark=none] coordinates { (0.00150,1.5249e-06) (0.00192,4.8609e-06) (0.00251,1.7016e-05) (0.00321,5.1793e-05) (0.00421,1.6991e-04) (0.00539,4.7977e-04) (0.00706,1.4175e-03) (0.00903,3.5515e-03) (0.01183,8.9079e-03) (0.01513,1.8590e-02) (0.01983,3.6606e-02) (0.02536,5.9087e-02) (0.03323,8.4977e-02) (0.04249,1.0246e-01) (0.05569,1.1248e-01) (0.07121,1.2285e-01) (0.09333,1.7624e-01) (0.11934,4.8754e-01) (0.15640,1.0000e+00) (0.20000,1.0000e+00) };
  \addplot[blue!70!black,mark=square*,mark options={fill=white}] coordinates {(0.005,0.0016666666666666668) (0.01,0.004) (0.02,0.015) (0.03,0.04) (0.05,0.1425) (0.07,0.3375) (0.1,0.5675) (0.13,0.7775)};
  \addlegendentry{[[90,4,10]]}
  \addplot[blue!70!black,dashed,mark=none] coordinates { (0.00150,1.5515e-06) (0.00192,4.9629e-06) (0.00251,1.7459e-05) (0.00321,5.3455e-05) (0.00421,1.7681e-04) (0.00539,5.0404e-04) (0.00706,1.5092e-03) (0.00903,3.8400e-03) (0.01183,9.8381e-03) (0.01513,2.1030e-02) (0.01983,4.2773e-02) (0.02536,7.1523e-02) (0.03323,1.0761e-01) (0.04249,1.3578e-01) (0.05569,1.5671e-01) (0.07121,1.7764e-01) (0.09333,2.5797e-01) (0.11934,6.8099e-01) (0.15640,1.0000e+00) (0.20000,1.0000e+00) };
  \addplot[green!60!black,mark=triangle*,mark options={fill=white}] coordinates {(0.001,0.0003333333333333333) (0.002,0.0016666666666666668) (0.005,0.0023333333333333335) (0.01,0.009333333333333334) (0.02,0.06) (0.03,0.142) (0.05,0.36) (0.07,0.5725) (0.1,0.875) (0.13,0.97)};
  \addlegendentry{[[90,10,8]]}
  \addplot[green!60!black,dashed,mark=none] coordinates { (0.00150,1.3789e-05) (0.00192,3.5417e-05) (0.00251,9.8694e-05) (0.00321,2.4686e-04) (0.00421,6.6262e-04) (0.00539,1.5868e-03) (0.00706,4.0080e-03) (0.00903,8.9500e-03) (0.01183,2.0525e-02) (0.01513,4.1100e-02) (0.01983,8.1391e-02) (0.02536,1.3894e-01) (0.03323,2.2453e-01) (0.04249,3.1292e-01) (0.05569,4.0537e-01) (0.07121,4.8232e-01) (0.09333,6.0341e-01) (0.11934,9.1765e-01) (0.15640,1.0000e+00) (0.20000,1.0000e+00) };
  \addplot[orange!80!black,mark=diamond*,mark options={fill=white}] coordinates {(0.03,0.0075) (0.05,0.09333333333333334) (0.07,0.22666666666666666) (0.1,0.6633333333333333) (0.13,0.9566666666666667)};
  \addlegendentry{[[170,10,8]]}
  \addplot[orange!80!black,dashed,mark=none] coordinates { (0.00150,3.4829e-08) (0.00192,9.3921e-08) (0.00251,2.8043e-07) (0.00321,7.6033e-07) (0.00421,2.2873e-06) (0.00539,6.2546e-06) (0.00706,1.9033e-05) (0.00903,5.2702e-05) (0.01183,1.6298e-04) (0.01513,4.5871e-04) (0.01983,1.4447e-03) (0.02536,4.1234e-03) (0.03323,1.3064e-02) (0.04249,3.6772e-02) (0.05569,1.0983e-01) (0.07121,2.7188e-01) (0.09333,6.0582e-01) (0.11934,9.0048e-01) (0.15640,7.1213e-01) (0.20000,1.9662e-01) };
  \addplot[violet!70!black,mark=pentagon*,mark options={fill=white}] coordinates {(0.005,0.003) (0.01,0.004) (0.02,0.011666666666666667) (0.03,0.055) (0.05,0.17333333333333334) (0.07,0.38666666666666666) (0.1,0.7866666666666666) (0.13,0.96)};
  \addlegendentry{[[170,10,6]]}
  \addplot[violet!70!black,dashed,mark=none] coordinates { (0.00150,3.0497e-05) (0.00192,6.2434e-05) (0.00251,1.3639e-04) (0.00321,2.7537e-04) (0.00421,5.8989e-04) (0.00539,1.1643e-03) (0.00706,2.4166e-03) (0.00903,4.6007e-03) (0.01183,9.0885e-03) (0.01513,1.6372e-02) (0.01983,3.0059e-02) (0.02536,5.0091e-02) (0.03323,8.3583e-02) (0.04249,1.2740e-01) (0.05569,1.9549e-01) (0.07121,2.8815e-01) (0.09333,4.7602e-01) (0.11934,9.1886e-01) (0.15640,1.0000e+00) (0.20000,1.0000e+00) };
  \addplot[cyan!70!black,mark=+,mark options={fill=white}] coordinates {(0.005,0.0005) (0.01,0.007) (0.02,0.016666666666666666) (0.03,0.0425) (0.05,0.11666666666666667) (0.07,0.24) (0.1,0.4866666666666667) (0.13,0.6066666666666667)};
  \addlegendentry{[[170,2,6]]}
  \addplot[cyan!70!black,dashed,mark=none] coordinates { (0.00150,1.8862e-05) (0.00192,3.8754e-05) (0.00251,8.5088e-05) (0.00321,1.7280e-04) (0.00421,3.7323e-04) (0.00539,7.4368e-04) (0.00706,1.5640e-03) (0.00903,3.0224e-03) (0.01183,6.0932e-03) (0.01513,1.1225e-02) (0.01983,2.1218e-02) (0.02536,3.6439e-02) (0.03323,6.2973e-02) (0.04249,9.8855e-02) (0.05569,1.5476e-01) (0.07121,2.2598e-01) (0.09333,3.4633e-01) (0.11934,5.5783e-01) (0.15640,1.0000e+00) (0.20000,1.0000e+00) };
  \addplot[magenta!70!black,mark=x,mark options={fill=white}] coordinates {(0.002,0.0005) (0.005,0.0035) (0.01,0.01) (0.02,0.03) (0.03,0.0925) (0.05,0.24666666666666667) (0.07,0.4166666666666667) (0.1,0.78) (0.13,0.9666666666666667)};
  \addlegendentry{[[170,10,7]]}
  \addplot[magenta!70!black,dashed,mark=none] coordinates { (0.00150,3.0861e-05) (0.00192,7.0391e-05) (0.00251,1.7237e-04) (0.00321,3.8395e-04) (0.00421,9.0931e-04) (0.00539,1.9486e-03) (0.00706,4.3725e-03) (0.00903,8.8088e-03) (0.01183,1.8158e-02) (0.01513,3.3263e-02) (0.01983,6.0415e-02) (0.02536,9.6648e-02) (0.03323,1.4838e-01) (0.04249,2.0224e-01) (0.05569,2.6509e-01) (0.07121,3.3338e-01) (0.09333,4.7443e-01) (0.11934,8.8374e-01) (0.15640,1.0000e+00) (0.20000,1.0000e+00) };
  \addplot[black,mark=diamond*] coordinates {(0.01,0.0006666666666666666) (0.012,0.00025) (0.015,0.00025) (0.02,0.001125) (0.025,0.00225) (0.03,0.0025) (0.035,0.004833333333333334) (0.04,0.00775) (0.045,0.014666666666666666) (0.05,0.0275) (0.07,0.035) (0.1,0.3075) (0.13,0.5825)};
  \addlegendentry{[[144,12,12]]}
  \addplot[black,dashed,mark=none] coordinates { (0.00150,4.2168e-09) (0.00192,1.7249e-08) (0.00251,7.9527e-08) (0.00321,3.1138e-07) (0.00421,1.3501e-06) (0.00539,4.9241e-06) (0.00706,1.9344e-05) (0.00903,6.3049e-05) (0.01183,2.1246e-04) (0.01513,5.8412e-04) (0.01983,1.5744e-03) (0.02536,3.4254e-03) (0.03323,6.9796e-03) (0.04249,1.1847e-02) (0.05569,1.9705e-02) (0.07121,3.3045e-02) (0.09333,8.2504e-02) (0.11934,4.2173e-01) (0.15640,1.0000e+00) (0.20000,1.0000e+00) };
\end{loglogaxis}
\end{tikzpicture}
\caption{Decoding performance: long codes vs BB $[[144,12,12]]$. Markers: Monte Carlo (BP-OSD, $50$ iterations, OSD-CS order $1$); dashed lines: the fit $\tilde p_L(p)=p^{d/2}e^{\alpha+\beta p+\gamma p^2}$ on sub-threshold data, drawn over the full axis range for illustration.}
\label{fig:decode_long}
\end{figure*}
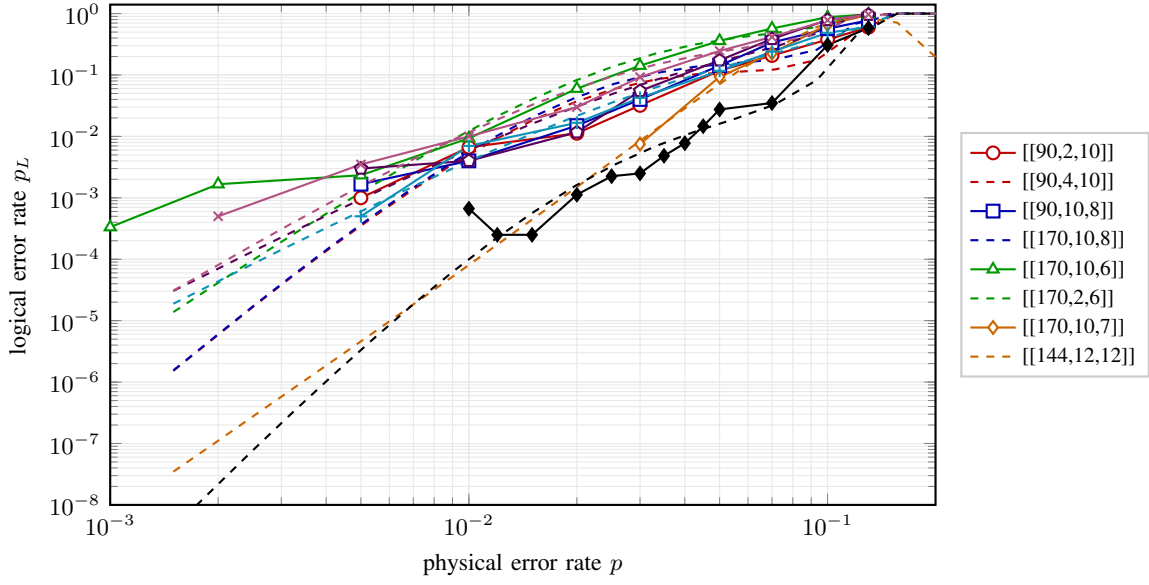
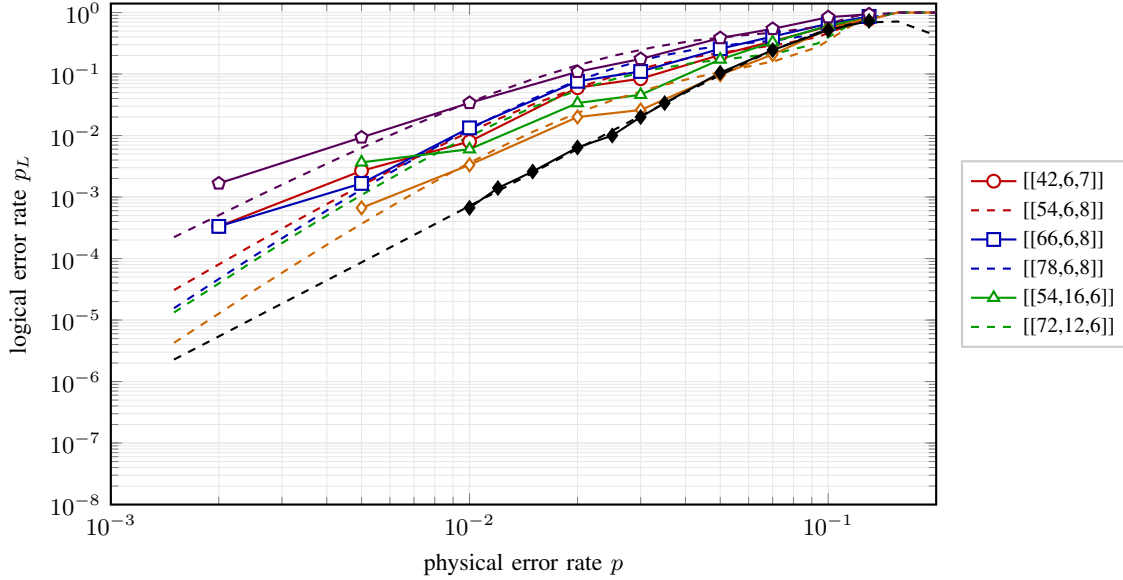
\begin{figure*}[t]
\centering
\begin{tikzpicture}
\begin{loglogaxis}[width=12.5cm,height=8.2cm,
  xlabel={physical error rate $p$}, ylabel={logical error rate $p_L$},
  xmin=0.001,xmax=0.2,ymin=1e-8,ymax=1.4,
  grid=both, grid style={gray!18},
  tick label style={font=\small}, label style={font=\small},
  legend style={font=\footnotesize, at={(1.03,0.5)}, anchor=west, draw=gray!40},
  legend cell align=left, thick, mark size=2.4pt]
  \addplot[red!75!black,mark=*,mark options={fill=white}] coordinates {(0.002,0.0003333333333333333) (0.005,0.0026666666666666666) (0.01,0.008) (0.02,0.06) (0.03,0.084) (0.05,0.205) (0.07,0.34) (0.1,0.595) (0.13,0.7825)};
  \addlegendentry{[[42,6,7]]}
  \addplot[red!75!black,dashed,mark=none] coordinates { (0.00150,3.0734e-05) (0.00192,7.0069e-05) (0.00251,1.7146e-04) (0.00321,3.8163e-04) (0.00421,9.0277e-04) (0.00539,1.9319e-03) (0.00706,4.3264e-03) (0.00903,8.6948e-03) (0.01183,1.7859e-02) (0.01513,3.2574e-02) (0.01983,5.8779e-02) (0.02536,9.3271e-02) (0.03323,1.4143e-01) (0.04249,1.8972e-01) (0.05569,2.4251e-01) (0.07121,2.9500e-01) (0.09333,3.9767e-01) (0.11934,6.8789e-01) (0.15640,1.0000e+00) (0.20000,1.0000e+00) };
  \addplot[blue!70!black,mark=square*,mark options={fill=white}] coordinates {(0.002,0.0003333333333333333) (0.005,0.0016666666666666668) (0.01,0.013333333333333334) (0.02,0.07625) (0.03,0.11) (0.05,0.2575) (0.07,0.4075) (0.1,0.6575) (0.13,0.865)};
  \addlegendentry{[[54,6,8]]}
  \addplot[blue!70!black,dashed,mark=none] coordinates { (0.00150,1.5574e-05) (0.00192,3.9814e-05) (0.00251,1.1022e-04) (0.00321,2.7354e-04) (0.00421,7.2629e-04) (0.00539,1.7174e-03) (0.00706,4.2618e-03) (0.00903,9.3258e-03) (0.01183,2.0797e-02) (0.01513,4.0358e-02) (0.01983,7.6632e-02) (0.02536,1.2501e-01) (0.03323,1.9083e-01) (0.04249,2.5157e-01) (0.05569,3.0758e-01) (0.07121,3.5308e-01) (0.09333,4.4570e-01) (0.11934,7.4980e-01) (0.15640,1.0000e+00) (0.20000,1.0000e+00) };
  \addplot[green!60!black,mark=triangle*,mark options={fill=white}] coordinates {(0.005,0.0036666666666666666) (0.01,0.006) (0.02,0.03375) (0.03,0.046) (0.05,0.1725) (0.07,0.33) (0.1,0.615) (0.13,0.8225)};
  \addlegendentry{[[66,6,8]]}
  \addplot[green!60!black,dashed,mark=none] coordinates { (0.00150,1.3258e-05) (0.00192,3.3709e-05) (0.00251,9.2594e-05) (0.00321,2.2775e-04) (0.00421,5.9711e-04) (0.00539,1.3916e-03) (0.00706,3.3844e-03) (0.00903,7.2387e-03) (0.01183,1.5652e-02) (0.01513,2.9360e-02) (0.01983,5.3362e-02) (0.02536,8.3249e-02) (0.03323,1.2078e-01) (0.04249,1.5289e-01) (0.05569,1.8287e-01) (0.07121,2.1590e-01) (0.09333,3.1374e-01) (0.11934,7.2472e-01) (0.15640,1.0000e+00) (0.20000,1.0000e+00) };
  \addplot[orange!80!black,mark=diamond*,mark options={fill=white}] coordinates {(0.005,0.0006666666666666666) (0.01,0.0033333333333333335) (0.02,0.02) (0.03,0.026) (0.05,0.0975) (0.07,0.2075) (0.1,0.54) (0.13,0.825)};
  \addlegendentry{[[78,6,8]]}
  \addplot[orange!80!black,dashed,mark=none] coordinates { (0.00150,4.2583e-06) (0.00192,1.0904e-05) (0.00251,3.0255e-05) (0.00321,7.5298e-05) (0.00421,2.0073e-04) (0.00539,4.7701e-04) (0.00706,1.1923e-03) (0.00903,2.6323e-03) (0.01183,5.9494e-03) (0.01513,1.1741e-02) (0.01983,2.2883e-02) (0.02536,3.8608e-02) (0.03323,6.2185e-02) (0.04249,8.8072e-02) (0.05569,1.2119e-01) (0.07121,1.6375e-01) (0.09333,2.7275e-01) (0.11934,6.8020e-01) (0.15640,1.0000e+00) (0.20000,1.0000e+00) };
  \addplot[violet!70!black,mark=pentagon*,mark options={fill=white}] coordinates {(0.002,0.0016666666666666668) (0.005,0.009333333333333334) (0.01,0.034) (0.02,0.10875) (0.03,0.176) (0.05,0.385) (0.07,0.5425) (0.1,0.84) (0.13,0.9375)};
  \addlegendentry{[[54,16,6]]}
  \addplot[violet!70!black,dashed,mark=none] coordinates { (0.00150,2.2261e-04) (0.00192,4.5058e-04) (0.00251,9.6853e-04) (0.00321,1.9187e-03) (0.00421,4.0014e-03) (0.00539,7.6539e-03) (0.00706,1.5198e-02) (0.00903,2.7482e-02) (0.01183,5.0508e-02) (0.01513,8.3722e-02) (0.01983,1.3696e-01) (0.02536,2.0024e-01) (0.03323,2.7979e-01) (0.04249,3.5081e-01) (0.05569,4.1799e-01) (0.07121,4.7525e-01) (0.09333,5.8224e-01) (0.11934,8.7963e-01) (0.15640,1.0000e+00) (0.20000,1.0000e+00) };
  \addplot[black,mark=diamond*] coordinates {(0.01,0.0006666666666666666) (0.012,0.0014) (0.015,0.0026) (0.02,0.0064) (0.025,0.01) (0.03,0.02) (0.035,0.033666666666666664) (0.05,0.105) (0.07,0.245) (0.1,0.52) (0.13,0.73)};
  \addlegendentry{[[72,12,6]]}
  \addplot[black,dashed,mark=none] coordinates { (0.00150,2.2799e-06) (0.00192,4.7827e-06) (0.00251,1.0815e-05) (0.00321,2.2733e-05) (0.00421,5.1549e-05) (0.00539,1.0869e-04) (0.00706,2.4747e-04) (0.00903,5.2404e-04) (0.01183,1.1994e-03) (0.01513,2.5523e-03) (0.01983,5.8674e-03) (0.02536,1.2502e-02) (0.03323,2.8595e-02) (0.04249,5.9925e-02) (0.05569,1.3122e-01) (0.07121,2.5359e-01) (0.09333,4.6796e-01) (0.11934,6.8003e-01) (0.15640,7.1493e-01) (0.20000,4.2365e-01) };
\end{loglogaxis}
\end{tikzpicture}
\caption{Decoding performance: short codes vs BB $[[72,12,6]]$. Markers: Monte Carlo (BP-OSD, $50$ iterations, OSD-CS order $1$); dashed lines: the fit $\tilde p_L(p)=p^{d/2}e^{\alpha+\beta p+\gamma p^2}$ on sub-threshold data, drawn over the full axis range for illustration.}
\label{fig:decode_short}
\end{figure*}

\section{Beyond the Cyclic Case: Rank Degeneracy at \texorpdfstring{$n=48$}{n=48}}\label{sec:rank48}
The polynomial framework of Sections~\ref{sec:poly}--\ref{sec:method} rests
on the cyclic, regular-action specialization $G=\Z_l$, $H=\{e\}$:
transposition is then a rank-preserving ring automorphism, and the
dimension formula collapses to the gcd form $k=2\deg\gcd(a,b,x^l-1)$ of
Proposition~\ref{prop:k}. This section states the main phenomena that
appear when these consequences fail, i.e.\ which effects are captured by
the polynomial picture and which are genuinely coset-theoretic. Proofs are
only sketched; the full census data and the structural analysis will be
developed in a companion paper.

Let $G$ be a finite group, $H\le G$ a subgroup of index $m=[G{:}H]$, and
$a,b\in\F_2[G/H]$ sums of $w_a$, $w_b$ coset permutations. The symmetric
coset-2BGA CSS code has checks $H_X=[\,A\mid B\,]$,
$H_Z=[\,B^{T}\mid A^{T}\,]$ on $n=2m$ qubits, with $H_XH_Z^{T}=0$ automatic
in characteristic $2$, and dimension
\begin{equation}\label{eq:kformula48}
k=2\bigl(m-\rank H_X\bigr)+\bigl(\rank H_X-\rank H_Z\bigr),
\end{equation}
the coset generalization of Proposition~\ref{prop:k}: the first term is the
generic (polynomial) contribution, while the rank defect
$\rank H_X-\rank H_Z$ vanishes identically in the cyclic case and can be
nonzero for nontrivial coset actions. We specialize to the weight-$8$
symmetric family $G=\mathrm{SmallGroup}(72,30)$, $H\cong C_3$, $m=24$,
$n=48$, $w_a=w_b=4$, enumerated exhaustively by computer into $123$
equivalence classes under $S_{48}$ coordinate permutation and the CSS
$X/Z$ swap.

\subsection{Three codes: benchmark, surpass, boundary}
Table~\ref{tab:48codes} places our construction against the two reference
points of Aydin--Tamo--Barg \cite{ATB}.

\begin{table}[t]
\centering
\caption{The $n=48$ coset-2BGA design space: benchmark, surpass, boundary}
\label{tab:48codes}
\renewcommand{\arraystretch}{1.2}
\begin{tabular}{@{}lccc@{}}
\toprule
 & this work $\mathcal C_{48}$ & ATB $\mathcal C_{72}$ & ATB $[[48,8,6]]$ \\
\midrule
parameters $[[n,k,d]]$ & $[[48,10,6]]$ & $[[48,10,6]]$ & $[[48,8,6]]$ \\
group $G$ & SG$(48,10)$ & SG$(72,30)$ & SG$(384,512)$ \\
group order $|G|$ & $\mathbf{48}$ & $72$ & $384$ \\
subgroup $H$ & $|H|=2$ & $C_3$ & $C_{16}$ \\
index $[G{:}H]$ & $24$ & $24$ & $24$ \\
$(w_a,w_b)$ & $(4,4)$ & $(4,4)$ & $(3,3)$ \\
check weight $w$ & $8$ & $8$ & $6$ \\
bit degree & $4$ & $4$ & $3$ \\
rate $k/n$ & $20.8\%$ & $20.8\%$ & $16.7\%$ \\
$kd/n$ & $1.25$ & $1.25$ & $1.0$ \\
transpose self-dual & verified & not verified & no \\
\bottomrule
\end{tabular}
\end{table}

\begin{construction}\label{cons:48}
Take $G_1=\mathrm{SmallGroup}(48,10)$, $H_1\le G_1$ with $|H_1|=2$,
$[G_1{:}H_1]=24$, and coset-support elements $a=[1,6,26,28]$,
$b=[1,6,7,8]$. Then $\rank H_X=\rank H_Z=19$ (GF(2) row reduction), hence
$k=10$, with exact distance $d=6$, constant row weight $8$, and constant
column weight $4$.
\end{construction}

The resulting code has the same public parameters
$(n,k,d,w)=(48,10,6,8)$ as the ATB code built from
$\mathrm{SmallGroup}(72,30)$; since the group order and the coset index
are invariants of a coset-2BGA construction, the two do not arise from
the same construction data.

\begin{theorem}\label{thm:minimal}
Within the $[G{:}H]=24$, $n=48$, weight-$8$ coset-2BGA framework,
Construction~\ref{cons:48} is the minimal-order group realization of
$(48,10,6)$: $n=2[G{:}H]=48$ forces $|G|=24|H|$, and $|H|=2$ is the
smallest nontrivial choice ($|H|=1$ degenerates to the regular action,
which the exhaustive 2BGA searches of \cite{ATB} show does not contain
these parameters). The group is $1/3$ smaller than the ATB realization of
the same parameters, and $7/8$ smaller than the $384$-element group
behind their $[[48,8,6]]$ code.
\end{theorem}

\subsection{The rank-degeneracy theorem}
\begin{theorem}\label{thm:rank-degeneracy}
In the weight-$8$ symmetric coset-2BGA family above, every code with
quantum distance $d=5$ satisfies
\[
\rank H_X=19,\qquad \rank H_Z=20,\qquad k=9.
\]
Consequently, no $[[48,10,5]]$ code exists in this family: distance $5$
and the generic dimension $k=10$ are mutually exclusive.
\end{theorem}
\begin{proof}[Proof (computer-assisted, exhaustive)]
The census partitions the family into $123$ equivalence classes; exactly
$7$ classes ($253$ member codes) have $d=5$ (Table~\ref{tab:d5-spectra}).
For each class both check ranks and the kernel weight spectra were
computed by exhaustive enumeration over all $\binom{48}{w}$ supports,
$w\le5$: in every case $\rank H_X=19$, $\rank H_Z=20$, hence $k=9$ by
\eqref{eq:kformula48}, and the X-side kernel contains weight-$5$ vectors
while both kernels are free of weight $\le4$. All remaining classes ---
$110$ with $d\le4$ and $6$ with $d=6$ --- have full ranks $19+19$ and
$k=10$, so rank loss is tied to distance exactly $5$, not to small or
large distance as such.
\end{proof}

\begin{table}[t]
\centering
\caption{The seven $[[48,9,5]]$ equivalence classes. $|{\rm Aut}|$ is the
automorphism group order; $N_5^X$ counts weight-$5$ X-type logicals.
All classes have $\rank H_X=19$, $\rank H_Z=20$, $d_X=5$, $d_Z\ge6$.}
\label{tab:d5-spectra}
\renewcommand{\arraystretch}{1.15}
\begin{tabular}{@{}cccccc@{}}
\toprule
class & members & $|{\rm Aut}|$ & $a$ & $b$ & $N_5^X$ \\
\midrule
6  & 80 & 48  & $[1,33,54,63]$ & $[1,5,9,10]$  & 24 \\
7  & 48 & 96  & $[1,22,53,72]$ & $[1,6,8,11]$  & 24 \\
8  & 48 & 96  & $[1,20,53,69]$ & $[1,5,7,10]$  & 24 \\
9  & 41 & 48  & $[1,35,38,62]$ & $[1,5,9,10]$  & 12 \\
10 & 16 & 192 & $[1,28,49,67]$ & $[1,6,8,11]$  & 24 \\
11 & 12 & 96  & $[1,21,46,68]$ & $[1,5,9,11]$  & 36 \\
12 & 8  & 48  & $[1,34,40,60]$ & $[1,5,7,10]$  & 12 \\
\bottomrule
\end{tabular}
\end{table}

\subsection{Structural features and outlook}
The seven classes share three features that single out the $d=5$ stratum:
every weight-$5$ X-kernel vector is a genuine logical operator (the Z
stabilizer contains no weight-$5$ element at all); $d_X=5$ while
$d_Z\ge6$ despite the transpose symmetry; and the minimum-logical counts
are quantized, $N_5^X\in\{12,24,36\}$, non-monotone in $|{\rm Aut}|$.

\begin{conjecture}\label{conj:module}
For this family, $\rank H_Z>\rank H_X$ iff the cyclic $G$-submodule
$\langle a\rangle\subseteq\F_2[G/H]$ is contained in a proper submodule
$U$ with $\dim\F_2[G/H]/U=1$, while $\langle b\rangle\not\subseteq U$.
Whenever this happens the code has $k=9$ and $d=5$; otherwise $k=10$ and
$d\in\{3,4,6\}$.
\end{conjecture}

For the present paper the message is twofold. First, the gcd dimension
formula of Proposition~\ref{prop:k} is the defect-zero case of the
general rank formula \eqref{eq:kformula48}, so the polynomial framework
knows exactly which degrees of freedom it has discarded. Second, the
module-embedding mechanism of Conjecture~\ref{conj:module} is the coset
analogue of the divisor condition $g\mid x^l-1$: both say that dimension
is controlled by where $a$ sits inside a module, and both turn this into
a pre-filter.

\section{Discussion}\label{sec:value}
\subsection{Theoretical significance}
Theorem~\ref{thm:weld} shows that the group-theoretic commuting-action
condition and the CRSS symplectic self-orthogonality are one equation read
two ways; under the cyclic specialization it becomes the polynomial
identity \eqref{eq:invol}. Proposition~\ref{prop:k} turns the quantum
dimension into a polynomial gcd, which enables pre-filtering and a
divisor-driven search that has no counterpart in the pure group
formulation. The framework is the cyclic, regular-action case of the
broader coset-based symplectic theory, and it makes precise when
group-theoretic constructions admit a polynomial (circulant) form.
Finally, the rank formula \eqref{eq:kformula48} shows that the gcd formula
of Proposition~\ref{prop:k} is the defect-zero case of a general two-term
law, and Theorem~\ref{thm:rank-degeneracy} exhibits the defect term in
action: at $n=48$ it forbids $[[48,10,5]]$ outright. The framework thus
states not only what it constructs but also what it provably cannot see.

\subsection{Practical significance}
The Magma pipeline (factor $x^l-1$, enumerate $g$, gcd pre-filter,
fingerprint dedup, two-level distance filter) is fully automated
and reproduces the short codes of Tables~\ref{tab:codes}--\ref{tab:newk2}
within seconds to minutes, while reaching
parameter combinations (e.g.\ large $k$) poorly covered by existing tables.

\subsection{Limitations and open problems}
The scope of the algebraic control established here is the univariate,
cyclic case. Extending it to the bivariate setting
$G=\Z_l\times\Z_m$ requires working in
$R_2=\F_2[x,y]/(x^l-1,y^m-1)$, where the single gcd of
Proposition~\ref{prop:k} must be replaced by a Gr\"obner-basis or
quotient-dimension computation; finding the right analogue is the natural
next step and would bring the BB codes of \cite{BB24} inside the same
pre-filtered search framework.

On the computational side, the exact distance evaluation of
Lemma~\ref{lem:exactd} is the bottleneck for large $l$. A two-level
estimator, in which a BP-OSD decoder proposes low-weight logical
candidates that are then certified exactly, would extend the reachable
lengths considerably. Finally, the performance statements in this paper
are made relative to the bicycle and bivariate/trivariate bicycle families
($kd^2/n$ comparisons, Section~\ref{sec:results}); a systematic comparison
against a comprehensive database of quantum LDPC codes, including
generalized bicycle and lifted-product families, remains future work.

Two open problems concern the coset direction of
Section~\ref{sec:rank48}. Conjecture~\ref{conj:module} predicts rank
degeneracy from a module-embedding condition; proving it, and unifying the
resulting criterion with the divisor pre-filter of
Section~\ref{sec:method}, would merge the cyclic and coset searches into a
single algebraic framework. Separately, the $[[48,10,6]]$ code of
Construction~\ref{cons:48} has weight-$8$ checks and therefore deeper
syndrome-extraction circuits than the weight-$6$ ATB $[[48,8,6]]$ code;
its circuit-level threshold behavior remains to be evaluated.

\section{Conclusion}\label{sec:concl}
We reformulated quantum bicycle LDPC codes in the symplectic polynomial
domain, showing that self-orthogonality is automatic
(Theorem~\ref{thm:weld}), the dimension is a polynomial gcd
(Proposition~\ref{prop:k}, Theorem~\ref{thm:window}), and the distance is
computed exactly via the Calderbank correspondence with the stabilizer
excluded (Lemma~\ref{lem:exactd}). The resulting divisor-driven search
recovers short codes at $n=46$, $66$, and $90$ with competitive
$kd^2/n$, including a $[[66,20,7]]_2$ code from a single degree-$10$
irreducible divisor of $x^{33}-1$ that attains $kd^2/n=14.85$, above the
bivariate bicycle code $[[144,12,12]]_2$, and reaches the large-$k$ regime
missed by group-theoretic searches. Stepping
outside the cyclic case, we showed that the same framework exposes its own
boundary: at $n=48$ the coset rank defect produces a rank-degeneracy theorem
excluding $[[48,10,5]]$ from the family, while a minimal-order $48$-element
group realizes $[[48,10,6]]$ within the coset-2BGA framework, improving on the
$72$-element construction of \cite{ATB}. A bivariate
extension, a sharp closed form for the kernel-family weight
$w_{\min}(a,b)$ of Lemma~\ref{lem:weightcap}, and a proof of the
module-embedding criterion are underway.

\section*{Acknowledgment}

The authors would like to thank Ruihu Li for the suggestions on our manuscript, which improved the manuscript significantly.
We also thank A. Canli (Quantavane) for an independent exact check that identified
one of the codes listed in Table~\ref{tab:codes} as $[[90,16,7]]_2$.
This work is supported by the National Natural Science Foundation of China under Grant No. U21A20428, Natural Science Foundation of Shaanxi under Grant No. 2025-JC-YBQN-070.


\bibliographystyle{IEEEtran}
\bibliography{refs}

\end{document}